\documentclass[
 reprint,
 longbibliography,
]{revtex4-2}

\pdfoutput=1

\usepackage[utf8]{inputenc}
\usepackage{hyperref}
\usepackage{amsmath}
\usepackage{amssymb}
\usepackage{amsthm}
\usepackage{physics}

\usepackage{tabularx}
\usepackage{multirow}
\usepackage{slashbox}
\usepackage{tikz}
\usetikzlibrary{quantikz}
\usepackage{xcolor}

\newcommand{\tsf}[1]{\textup{\textsf{#1}}}
\newcommand{\supp}{\mathrm{supp}}

\newcommand{\cnot}{\textsc{cnot}}
\newcommand{\cz}{\textsc{cz}}

\newtheorem{theorem}{Theorem}
\newtheorem{definition}{Definition}

\theoremstyle{definition}
\newtheorem{condition}{Condition}

\begin{document}

\title{Universal resource-efficient topological measurement-based quantum computation via color-code-based cluster states}

\author{Seok-Hyung Lee}
\affiliation{Department of Physics and Astronomy, Seoul National University, Seoul 08826, Republic of Korea}
\author{Hyunseok Jeong}
\email{h.jeong37@gmail.com}
\affiliation{Department of Physics and Astronomy, Seoul National University, Seoul 08826, Republic of Korea}

\begin{abstract}
Topological measurement-based quantum computation (MBQC) enables one to carry out universal fault-tolerant quantum computation via single-qubit Pauli measurements with a family of large entangled states called cluster states as resources.
Raussendorf's three-dimensional cluster states (RTCSs) based on the surface codes are mainly considered for topological MBQC.
In such schemes, however, the fault-tolerant implementation of the logical Hadamard, phase ($Z^{1/2}$), and $T$ ($Z^{1/4}$) gates which are essential for building up arbitrary logical gates has not been achieved to date without using state distillation, while the controlled-\textsc{not}~(\cnot) gate does not require it, to best of our knowledge.
State distillation generally consumes many ancillary logical qubits, thus it is a severe obstacle against practical quantum computing.
To solve this problem, we suggest an MBQC scheme via a family of cluster states called \textit{color-code-based cluster states} (CCCSs) based on the two-dimensional color codes instead of the surface codes.
We define logical qubits, construct elementary logical gates, and describe error correction schemes.
We show that all the logical Clifford gates including the \cnot, Hadamard, and phase gates can be implemented fault-tolerantly without state distillation, although the fault-tolerant $T$ gate still requires it.
We further prove that the minimal number of physical qubits per logical qubit in a CCCS is at most approximately 1.8 times smaller than the case of an RTCS.
We lastly show that the error threshold of MBQC via CCCSs for logical-$Z$ errors is 2.7--2.8\%, which is comparable to the value for RTCSs, assuming a simple error model where physical qubits have $X$-measurement or $Z$ errors independently with the same probability.
\end{abstract}

\maketitle

\section{Introduction}
\label{sec:introduction}

Three major theoretical challenges for quantum computation (QC) are \textit{universality}, \textit{fault-tolerance}, and \textit{resource-efficiency}. 
Universality indicates the ability of a quantum computer to initialize logical qubits to the computational basis state, perform any unitary gate, and measure them in the computational basis.
It is known that the controlled-\textsc{not}~(\cnot), Hadamard, and phase ($Z^{1/2}$) gates completely generate the Clifford group, and together with the $T$ ($Z^{1/4}$) gate, any unitary gate may be approximated to an arbitrary accuracy \cite{galindo2002information, nielsen2010quantum}.

To achieve fault-tolerance, various quantum error-correcting (QEC) codes have been proposed from simple codes with few physical qubits \cite{shor1995scheme, bennett1996mixed, laflamme1996perfect, calderbank1996good, steane1996multiple} to topological stabilizer codes defined on lattice structures of qubits allowing only local interactions which are easily scalable \cite{bombin2013topological}.
Several simple codes also have been demonstrated experimentally in assorted systems for small code distances \cite{cory1998experimental, chiaverini2004realization, schindler2011experimental, reed2012realization, nigg2014quantum, corcoles2015demonstration, kelly2015state, ofek2016extending, andersen2020repeated}.
Particularly, the \textit{surface codes} \cite{kitaev1997quantum, bravyi1998quantum, dennis2002topological, kitaev2003fault, fowler2009high, bombin2009quantum, fowler2012surface, terhal2015quantum}, a family of topological codes defined on two-dimensional (2D) lattices, have high \textit{error thresholds} up to about 12\% \cite{fowler2012surface}, thus they are one of the most promising candidates for fault-tolerant QC.
The \textit{2D color codes} is another family of topological codes \cite{bombin2006topological, fowler2011two, bombin2013topological, kesselring2018boundaries} which enable transversal \footnote{
    That a logical gate $U$ is transversal means that $U$ can be expressed as $U = U_1 U_2 \cdots$ such that a unitary operator $U_i$ acts only on the $i$th physical qubit for each logical qubit for all $i$'s. 
    For example, if $X_L := X_1 \cdots X_n$ for $[[n, 1, d]]$ code where $X_i$ is the $X$ operator on the $i$th physical qubit, $X_L$ is transversal. 
    Specific 2D color codes implement the logical Hadamard gate by the combination of the Hadamard gate on every physical qubit, and similarly for the logical phase gate \cite{bombin2006topological, fowler2011two}.
} implementation of the logical \cnot, Hadamard, and phase gates thanks to their self-duality.
Moreover, three-dimensional (3D) gauge color codes even allow transversal implementation of the logical $T$ gate as well as the Clifford gates, thus have been getting much attention recently \cite{bombin2007topological, bombin2007exact, bombin2015gauge, kubica2015universal, watson2015qudit, kubica2018three, bombin20182d, bombin2018transversal}.

Lastly, fault-tolerant QC typically requires enormous resource overheads, which makes it tough to realize it.
It is not only because a single logical qubit is composed of multiple physical qubits, but also because state distillation, which generally demand many ancillary logical qubits, is required for non-Clifford gates and sometimes for several Clifford gates to be fault-tolerant \cite{bravyi2005universal, fowler2009high, fowler2012surface, jones2013multilevel}, e.g., one round of a typical protocol to distill an $\ket{A_L} := \frac{1}{\sqrt{2}}\qty( \ket{0_L} + e^{i\pi/4}\ket{1_L} )$ state for the logical $T$ gate requires 15 ancillary logical qubits \cite{bravyi2005universal, raussendorf2007topological, fowler2012surface}.
It is therefore desirable to find QC schemes minimizing the need for state distillation.

\textit{Measurement-based QC} (MBQC) is an alternative of conventional circuit-based QC (CBQC), processed only by single-qubit Pauli measurements with a family of large entangled states called \textit{cluster states} as resources \cite{raussendorf2001one, raussendorf2003measurement, raussendorf2006fault, raussendorf2007fault, raussendorf2007topological, fowler2009topological}.
The ingredients for generating Cluster states are physical qubits initialized to the $X$ basis and controlled-$Z$ (\cz) gates on them, thus MBQC requires much fewer types of physical-level operations than typical CBQC.
The initial MBQC schemes via cluster states on 2D planes \cite{raussendorf2001one, raussendorf2003measurement} were universal but not fault-tolerant.
To achieve fault-tolerance, the space should be 3D; \textit{Raussendorf's 3D cluster states} (RTCSs) allow universal and fault-tolerant MBQC with topologically-encoded logical qubits \cite{raussendorf2006fault, raussendorf2007fault, raussendorf2007topological, fowler2009topological}.
Additionally, it was shown that it can tolerate imperfect preparation of cluster states such as qubits losses or failed \cz~gates \cite{barrett2010fault, whiteside2014upper, li2010fault}.
MBQC via cluster states is regarded as one of the suitable candidates for practical fault-tolerant QC, especially in optical systems \cite{nielsen2004optical, dawson2006noise, menicucci2006universal, devitt2009architectural, herrera2010photonic, li2010fault, fujii2010fault, myers2011coherent, yao2012experimental, gimeno2015from, li2015resource, omkar2020resource}.

MBQC via RTCSs is powerful from the point of view of universality and fault-tolerance, but has a significant drawback: the fault-tolerant implementation of the logical Hadamard, phase, and $T$ gates has not been achieved to date without using costly state distillation, while the \cnot~gate does not require it \cite{raussendorf2006fault, raussendorf2007fault, raussendorf2007topological, fowler2009topological}, to best of our knowledge.
Fault-tolerant non-Clifford gates (including the $T$ gate) demand state distillation for most topological QEC codes, but it is uncommon that it is needed even for some Clifford gates (the Hadamard and phase gates), which is fatal to resource-efficiency.
To solve this problem, we consider recent studies on generalizing the relation between RTCSs and surfaces codes to any Calderbank-Steane-Shor (CSS) codes \cite{bolt2016foliated, bolt2018decoding} and later to any stabilizer codes \cite{brown2020universal}.
We also consider Ref. \cite{fowler2011two} on 2D color code computation introducing ``defects'' for logical operations since a similar approach on the surface codes leads to the protocol for MBQC via RTCSs.
Motivated by these works, we propose a new MBQC scheme via a family of cluster states based on the 2D color codes instead of the surface codes, called \textit{color-code-based cluster states} (CCCSs).

Throughout this paper, we show that MBQC via CCCSs is a competitive candidate on realistic QC, regarding the three challenges mentioned at the very first.
We describe the fault-tolerant implementation of the logical \cnot, Hadamard, and phase gates without state distillation, thus show that arbitrary Clifford gates do not require it to be fault-tolerant unlike MBQC via RTCSs, although non-Clifford gates still demand it.
We then prove that it requires a smaller amount of physical qubits per logical qubits than the case of RTCSs, which makes it resource-efficient even more.
Finally, we show that they have a similar level of fault-tolerance by comparing their error thresholds.

This paper is structured as follows. 
In Sec.~\ref{sec:cluster_states_and_MBQC}, we review the concept of cluster states and the general process of MBQC.
In Sec.~\ref{sec:color_code_based_cluster_states}, we construct CCCSs and describe their properties, especially about their stabilizers called \textit{correlation surfaces}.
In Sec.~\ref{sec:MBQC_via_color_code_based_cluster_states}, we define logical qubits and suggest the schemes for their initialization and measurements, elementary logical gates, and state injection for state distillation.
In Sec.~\ref{sec:error_correction}, we present the methods for error correction.
In Sec.~\ref{sec:calculations}, we calculate resource overheads and error thresholds of MBQC via CCCSs and compare them with the results for RTCSs.
We conclude with final remarks in Sec.~\ref{sec:conclusion}.


\section{Cluster states and measurement-based quantum computation}
\label{sec:cluster_states_and_MBQC}

\begin{figure}[t!]
	\centering
	\includegraphics[width=\columnwidth]{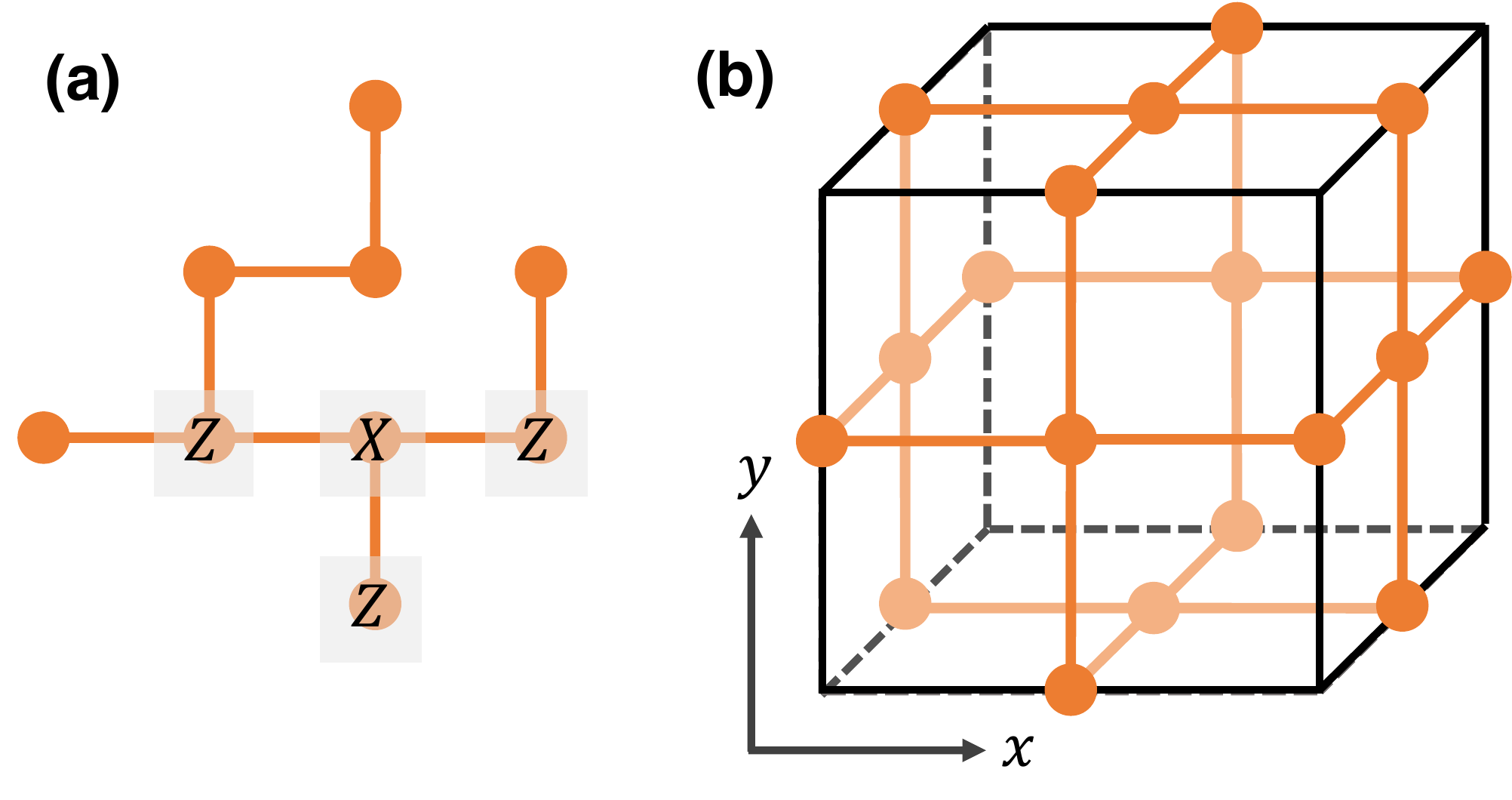}
	\caption{
    	Examples of cluster states. 
    	Orange dots and lines indicate vertices and edges of the graphs, respectively.
    	To construct a cluster state, qubits initialized to the $\ket{+} = \frac{1}{\sqrt{2}} \qty( \ket{0} + \ket{1} )$ states are placed on the vertices, then a \cz~gate is applied on the qubits connected by each edge.
    	(a) A cluster state on a simple graph.
    	The presented ``$XZZZ$'' operator indicates an example of a stabilizer generator given in Eq.~\eqref{eq:cluster_state_SG}.
    	(b) A unit cell of a Raussendorf's 3D cluster state (RTCS).
    	A vertex is located on each edge and face of the cell.
	}
	\label{fig:cluster_states}
\end{figure}

To define a cluster state, we consider a graph $G=(V,~E)$, where $V$ and $E$ are the sets of vertices and edges, respectively.
The cluster state $\ket{G}$ is constructed by attaching a qubit to every vertex in $V$, initializing them to the $\ket{+}$ states where $\ket{\pm} := \qty( \ket{0} \pm \ket{1} )/\sqrt{2}$, then applying a controlled-$Z$ (\cz) gate on every pair of qubits connected by an edge.

The constructed cluster state has a stabilizer generator (SG) $S(v)$ for each vertex $v \in V$ defined as
\begin{align}
    S(v) := X(v) \prod_{v' \in \mathrm{adj}(v)} Z(v'),
    \label{eq:cluster_state_SG}
\end{align}
where $\mathrm{adj}(v) := \qty{ v' \in V \mid (v, v') \in E }$ is the set of adjacent vertices of $v$ and $X(v)$ and $Z(v)$ are the $X$ and $Z$ operators, respectively, on the qubit at the vertex $v$ denoted by $Q(v)$.
In other words, $g\ket{G} = \ket{G}$ holds for all $g \in \mathcal{S}$, where $\mathcal{S}$ is the stabilizer group generated by $\qty{S(v) | v \in V}$.
We say that $S(v)$ is \textit{around} $v$ or $Q(v)$, called its \textit{center vertex} or \textit{qubit}, respectively.
Examples of cluster states are presented in Fig.~\ref{fig:cluster_states}.

For MBQC, modified versions of cluster states are used, where some qubits do not need to be initialized to the $\ket{+}$ states.
$S(v)$ where $Q(v)$ is each of such qubits is then no longer a stabilizer, but others still remain as SGs.

General MBQC via a cluster state to implement a quantum circuit is processed through the following three steps \cite{raussendorf2001one, raussendorf2003measurement, raussendorf2007fault, raussendorf2007topological, fowler2009topological}:
\begin{enumerate}
    \item 
        \textit{Preparation.} 
        For a given graph $G(V, E)$, a qubit is attached to each vertex.
        $Q(V)$ is divided into three subsets: the input qubits $Q_{\mathrm{IN}}$, the output qubits $Q_{\mathrm{OUT}}$, and the others.
        $Q_{\mathrm{OUT}} = \emptyset$ if the desired circuit does not produce any output state or ends with measurements.
        The input logical states are prepared in $Q_{\mathrm{IN}}$.
        All qubits except those in $Q_\mathrm{IN}$ are initialized to the $\ket{+}$ states.
        A \cz~gate is then applied on every pair of qubits connected by an edge.
    \item
        \textit{Measurement.}
        For each physical qubit except those in $Q_\mathrm{OUT}$, a single-qubit Pauli measurement, selected by a \textit{measurement pattern} with a classical computer, is performed.
        The measurement pattern is determined by the desired circuit.
        Let the measurement results be $M$.
        If possible, errors in $M$ are corrected by decoding the \textit{parity-check} outcomes.
    \item
        \textit{Obtaining the results.}
        The output logical state is obtained from $Q_{\mathrm{OUT}}$ up to logical Pauli operators called \textit{byproduct operators} determined by $M$.
        If $Q_{\mathrm{OUT}} = \emptyset$, the results of the final logical measurements are determined by $M$.
\end{enumerate}

Note that the preparation and measurement steps may be performed simultaneously; after a qubit $q$ and its neighbors are prepared and \cz~gates are applied on them, it is allowed that $q$ is measured before the other qubits are prepared.
One of the spatial axes may be regarded as the \textit{simulating time axis}, or simply the \textit{time axis}, along which qubits are prepared and then measured in order.
It is thus possible to minimize the number of unmeasured physical qubits in a moment by measuring each qubit as soon as possible after preparing it.

The cores of MBQC are the structure of the cluster state and the measurement pattern.
We illustrate them one by one over the next two sections.


\section{Color-code-based cluster states}
\label{sec:color_code_based_cluster_states}

In this section, we define color-code-based cluster states and describe their properties.
Based on the work on the foliation of CSS codes \cite{bolt2016foliated}, we consider a particular family of cluster states derived from a 2D color-code lattice, called \textit{color-code-based cluster states} (CCCSs).

\subsection{Two-dimensional color-code lattices}
\label{subsec:2d_color_code_lattices}

\begin{figure}[t!]
	\centering
	\includegraphics[width=\columnwidth]{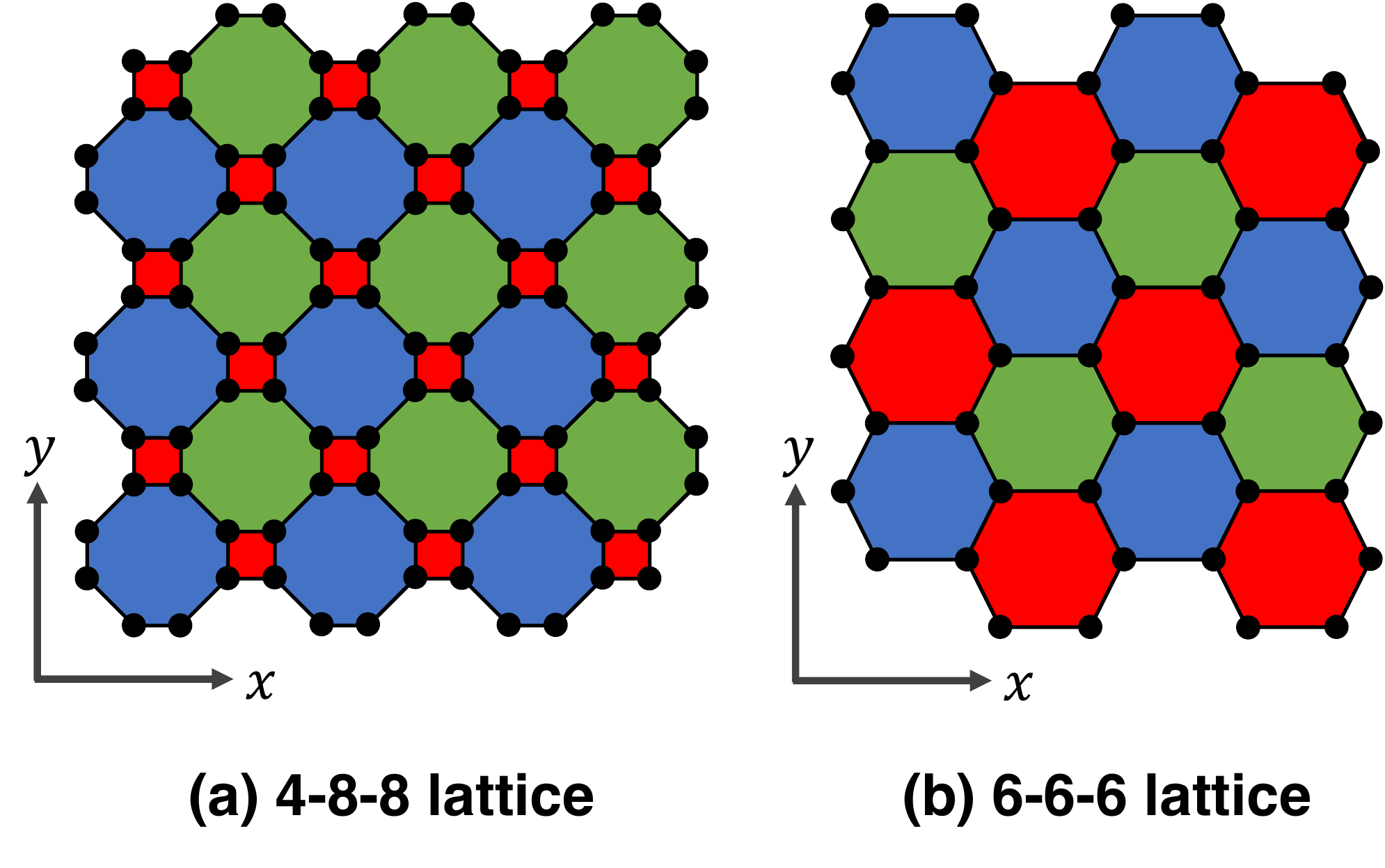}
	\caption{
    	Two typical examples of color-code lattices: the (a) 4-8-8 and (b) 6-6-6 lattices. 
    	The lattices are 3-valent and have 3-colorable faces.
	}
	\label{fig:color_code_lattices}
\end{figure}

We first present 2D color-code lattices on which CCCSs are based.
We consider a lattice $\mathcal{L}_{2D}$ on a 2D plane which is 3-valent and has 3-colorable faces; namely, three edges meet at each vertex and one of the three colors (red, green, or blue) is assigned to each face in such a way that neighboring faces have different colors. 
Note that each edge, called \textit{link}, is also colorable with the color of the faces it connects.
Two typical examples (4-8-8 and 6-6-6) of such lattices are shown in Fig.~\ref{fig:color_code_lattices}.
In the original 2D color codes, a qubit is attached to each vertex and two SGs ($X$- and $Z$-type) correspond to each face.
Details on the codes are described in Ref. \cite{bombin2006topological, bombin2013topological}.


\begin{figure}[t!]
	\centering
	\includegraphics[width=\columnwidth]{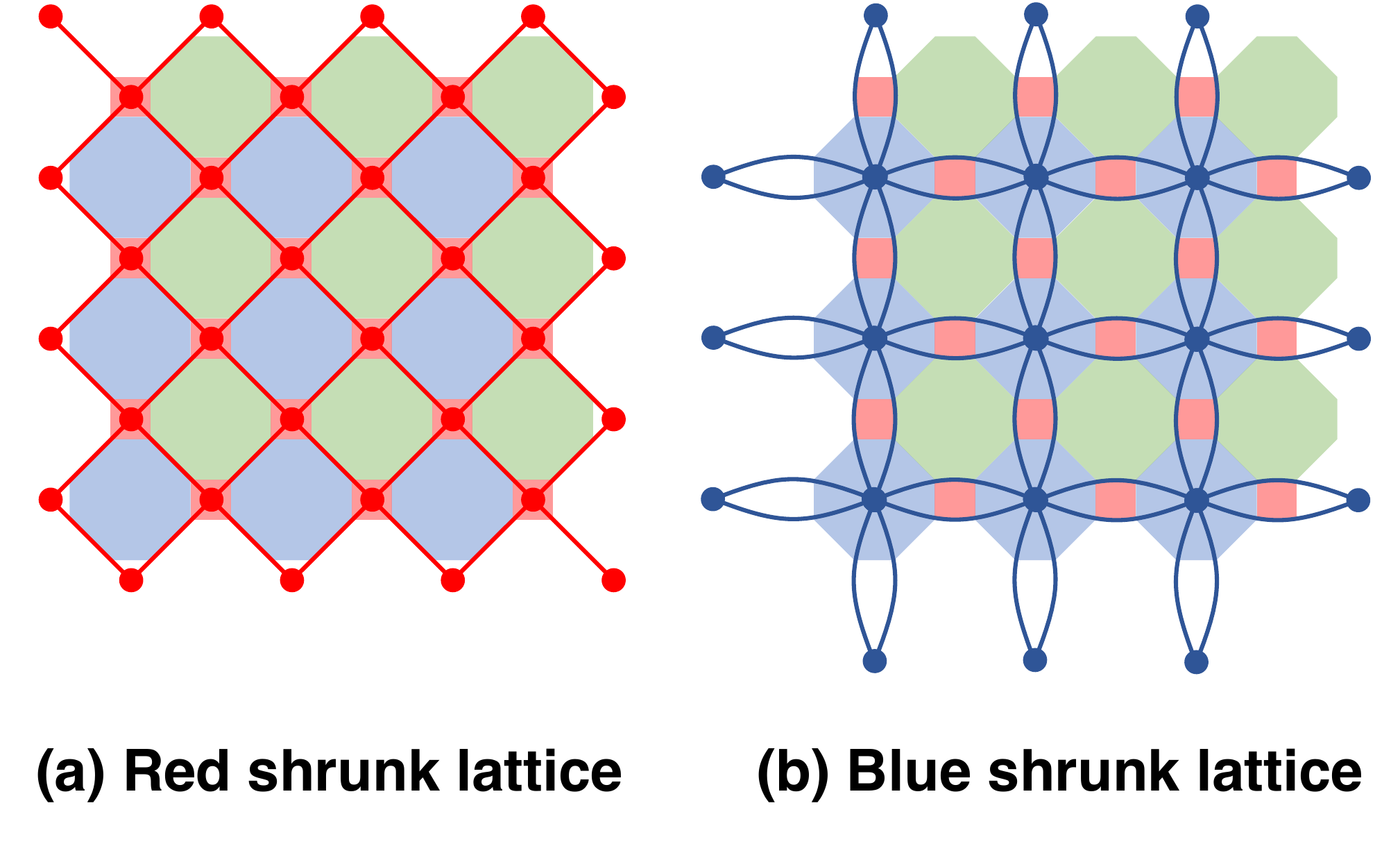}
	\caption{
    	(a) Red and (b) blue shrunk lattices of the 4-8-8 color-code lattice.
    	Red or blue dots (lines) indicate their vertices (edges), which corresponds to red or blue faces (links) of the original lattices.
	}
	\label{fig:color_code_shrunk_lattices}
\end{figure}

Regarding a color-code lattice $\mathcal{L}_{2D}$, three \textit{shrunk lattices} are defined, one for each color by shrinking all the faces of that color, as shown in Fig.~\ref{fig:color_code_shrunk_lattices}.
For example, in the red shrunk lattice, each vertex corresponds to a red face in $\mathcal{L}_{2D}$ and each face corresponds to a blue or green face in $\mathcal{L}_{2D}$. 
Edges of the red shrunk lattice then correspond to red links in $\mathcal{L}_{2D}$. 
The blue and green shrunk lattices are also defined analogously.

\subsection{Construction of color-code-based cluster states}
\label{subsec:basic_structures}

\begin{figure*}[t!]
	\centering
	\includegraphics[width=\textwidth]{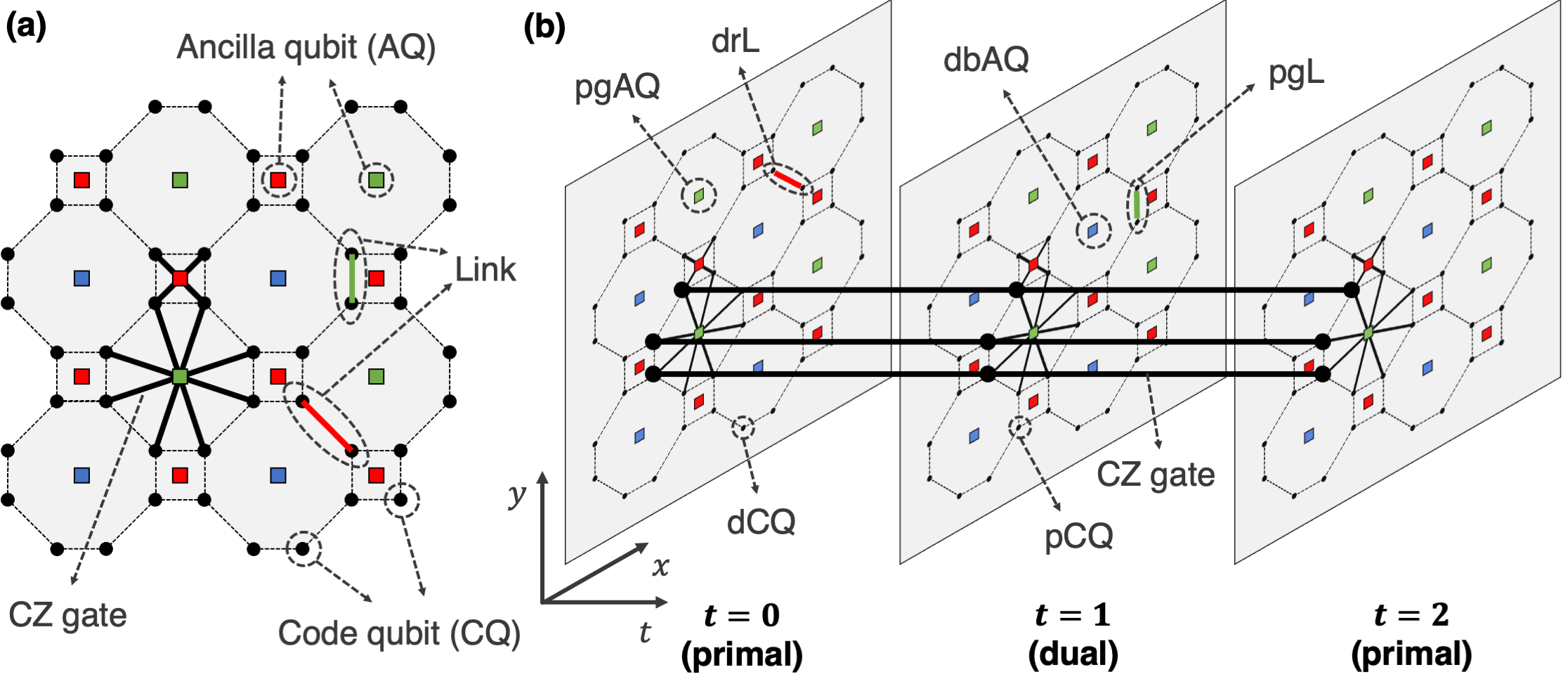}
	\caption{
    	Structure of a color-code-based cluster state (CCCS) based on the 4-8-8 color-code lattice $\mathcal{L}_{2D}$. 
    	(a) Structure of a single layer.
    	Each black circle is a code qubit (\tsf{CQ}) located at a vertex of $\mathcal{L}_{2D}$.
    	Each colored square is an ancilla qubit (\tsf{AQ}) with that color, located at the center of a face of $\mathcal{L}_{2D}$ with that color.
    	Each \tsf{AQ} is connected with surrounding \tsf{CQ}s by edges (\cz~gates), some of which are drawn as black solid lines.
        Two adjacent \tsf{CQ}s are connected by a link, some of which are drawn as colored lines.
    	(b) Stack of multiple identical layers along the simulating time axis.
    	Each pair of two \tsf{CQ}s adjacent along the time axis is connected by an edge, some of which are presented as black solid lines.
        One of the primalities (``primal'' and ``dual'') is alternatively assigned to each layer.
    	An \tsf{AQ} (a \tsf{CQ} or link) is primal (dual) if it is in a primal layer, and vice versa for a dual layer.
    	Labels of some elements defined in Sec.~\ref{subsec:basic_structures} are shown. 
	}
	\label{fig:CCCS_structure}
\end{figure*}

The graph $G$ for a CCCS based on a color-code lattice $\mathcal{L}_{2D}$ has a 3D structure composed of multiple identical 2D layers stacked along the simulating time ($t$) axis.
The layer of $t=t_0$ is referred to as the \textit{$t_0$-layer}.

The structure of each layer is originated from $\mathcal{L}_{2D}$, as illustrated in Fig.~\ref{fig:CCCS_structure}(a) for the case of the 4-8-8 lattice.
Each vertex in the layer is located at either a vertex of $\mathcal{L}_{2D}$ or the center of a face of $\mathcal{L}_{2D}$; the corresponding qubit is called a \textit{code qubit} (\tsf{CQ}) or an \textit{ancilla qubit} (\tsf{AQ}), respectively.
Each \tsf{AQ} is colorable with the color of the corresponding face in $\mathcal{L}_{2D}$.
For each face in $\mathcal{L}_{2D}$, the layer has an edge connecting the corresponding \tsf{AQ} and each surrounding \tsf{CQ}, on which a \cz~gate is applied. 
Each pair of \tsf{CQ}s connected by a link in $\mathcal{L}_{2D}$ is called \textit{link} here as well.
Note that links are not edges of $G$.

Next, we stack multiple identical layers along the time axis as shown in Fig.~\ref{fig:CCCS_structure}(b).
Every pair of \tsf{CQ}s adjacent along the time axis is connected by an edge in $G$.
The vertices (\tsf{CQ}s and \tsf{AQ}s) and edges (between \tsf{CQ}s and \tsf{AQ}s in the same layer and between \tsf{CQ}s in the adjacent layers) constructed above finally complete the graph $G$ of the cluster state.

We assign each layer, qubit or link a ``primality'': either \textit{primal} or \textit{dual}.
Each layer is primal (dual) if it has an even (odd) time.
An \tsf{AQ} is primal (dual) if it is in a primal (dual) layer, while a \tsf{CQ} or link is primal (dual) if it is in a dual (primal) layer.
We label each qubit or link in an abbreviated form with its primality (``\tsf{p}'' for primal and ``\tsf{d}'' for dual), color (``\tsf{r}'' for red, ``\tsf{g}'' for green, and ``\tsf{b}'' for blue; omitted for \tsf{CQ}s), and type (``\tsf{AQ},'' ``\tsf{CQ},'' and ``\tsf{L}'' for a link).
For example, a \tsf{pgAQ} means a primal green ancilla qubit.
We also frequently use ``\tsf{c}'' instead of a specific color (\tsf{r}, \tsf{g}, or \tsf{b}) for a variable on colors.

\subsection{Stabilizer generators}
\label{subsec:stabilizer_generators}

\begin{figure}[t!]
	\centering
	\includegraphics[width=\columnwidth]{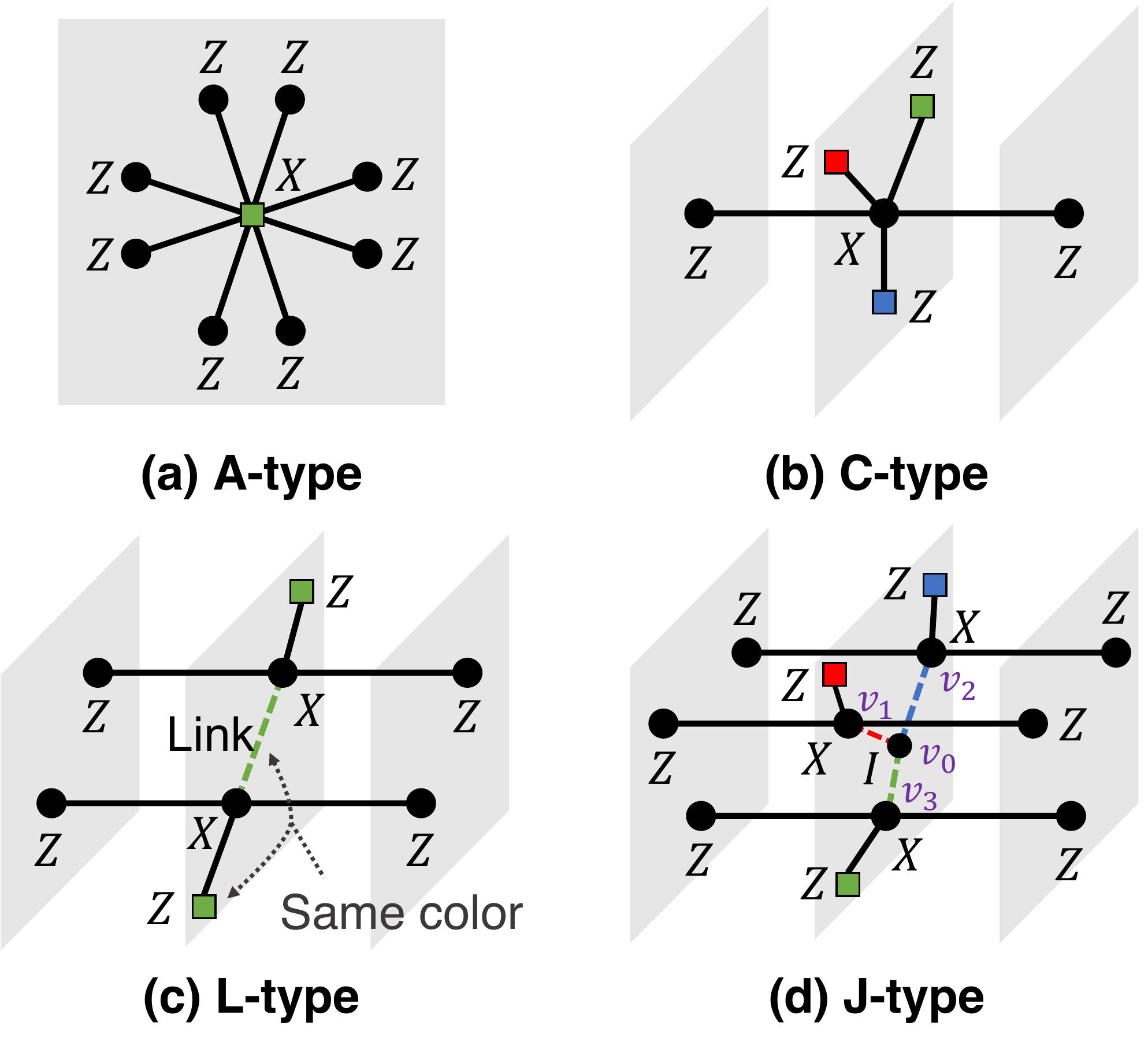}
	\caption{
    	Four types of stabilizer generators (SGs) in a CCCS defined in Definition \ref{def:AC_type_SGs}--\ref{def:I_type_SGs}: (a) A-, (b) C-, (c) L-, and (d) J-type.
    	Each grey square indicates a layer.
    	An SG of each type is the tensor product of the marked $X$ or $Z$ operators on the qubits.
	}
	\label{fig:cluster_state_SGs}
\end{figure}

We now present stabilizer generators (SGs) of a CCCS.
Remark that, for each vertex $v$ in $G$, $S(v)$ given in Eq.~\eqref{eq:cluster_state_SG} is a SG if $Q(v)$ is initialized to $\ket{+}$.
We define \textit{A-} and \textit{C-type SGs} shown in Fig.~\ref{fig:cluster_state_SGs}(a) and (b) as follows.

\begin{definition}[\textbf{A- and C-type SGs}]
    An \textup{A-} or \textup{C-type SG} is the SG given in Eq.~\eqref{eq:cluster_state_SG} around an \tsf{AQ} or a \tsf{CQ}, respectively.
    \label{def:AC_type_SGs}
\end{definition}

The support \footnote{The support of an operator $O$, written as $\supp(O)$, is the set of qubits on which $O$ applies non-trivially.} of a C-type SG is distributed in three adjacent layers, while that of an A-type SG is contained in a layer.

Although these two types of SGs completely generate the stabilizer group, we need another two types of SGs: \textit{L-} and \textit{J-type SGs} in Fig.~\ref{fig:cluster_state_SGs}(c) and (d).

\begin{definition}[\textbf{L-type SG}]
    The \textup{L-type SG around a link $l$} is the product of two C-type SGs whose center qubits constitute $l$.
    \label{def:L_type_SGs}
\end{definition}

\begin{definition}[\textbf{J-type SG}]
    Let $S_i := S\qty(v_i)$ for each $i \in \qty{ 0, 1, 2, 3 }$ be a C-type SG such that $\qty( v_0, v_1 )$, $\qty( v_0, v_2 )$, and $\qty( v_0, v_3 )$ are links with different colors.
    $S_I := S_1 S_2 S_3$ is then the \textup{J-type SG around the \tsf{CQ} $Q\qty(v_0)$}.
    \label{def:I_type_SGs}
\end{definition}


A-, L-, and J-type SGs together generate the stabilizer group over-completely.
To see this, regarding a J-type SG $S_I$, we consider an L-type SG $S_{Li} := S_0 S_i$ for each $i \in \qty{ 0, 1, 2, 3 }$, where $S_i$'s are defined in Definition \ref{def:I_type_SGs}.
Then $S_0 = S_{L1} S_{L2} S_{L3} S_I$ holds, thus any C-type SG can be written as the product of L- and J-type SGs.

Let the $P$-support $\supp_P(O)$ of a multi-qubit Pauli operator $O$ for $P \in \qty{ X, Y, Z }$ be the subset of $\supp(O)$ corresponding to the $P$ operator.
Note that, for every SG regardless of its type, qubits in its $X$- and $Z$-support always have different primalities.

\subsection{Shrunk lattices and correlation surfaces}
\label{subsec:shrunk_lattices_corr_surfaces}

Almost every discussion from now on is symmetric between the two primalities.
Thus, throughout the rest of this paper, we frequently discuss only one of them, which implies that the other side can be treated similarly.

We now construct the shrunk lattices of a CCCS, which are analogous to those of the 2D color codes in Fig.~\ref{fig:color_code_shrunk_lattices}. 
We then define correlation surfaces \cite{raussendorf2006fault, raussendorf2007fault, raussendorf2007topological} within each shrunk lattice, through which logical gates are built for MBQC.

\begin{figure}[t!]
	\centering
	\includegraphics[width=\columnwidth]{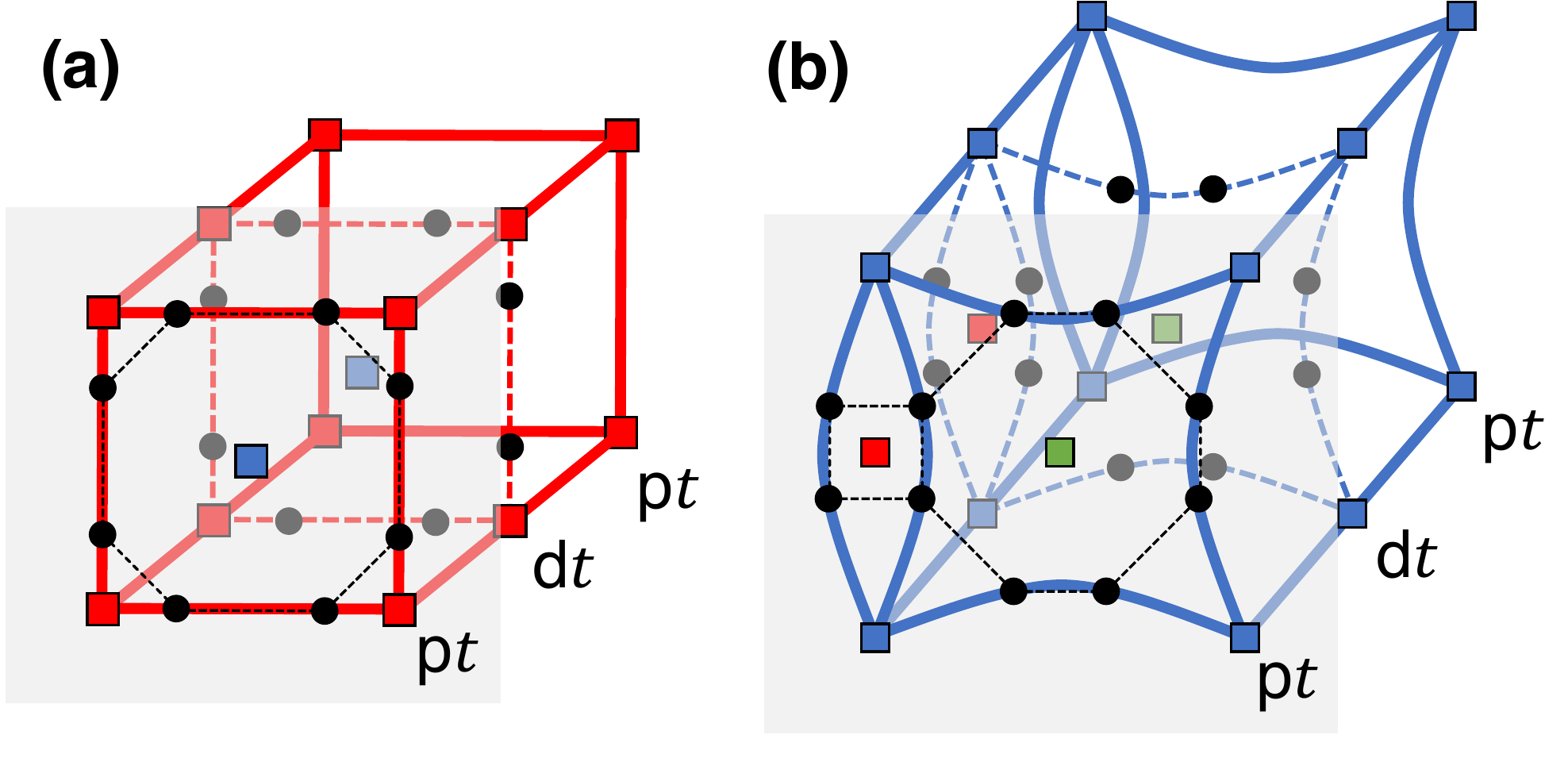}
	\caption{
    	Unit cells of the primal shrunk lattices of an 4-8-8 CCCS: (a) a blue cell in the primal red shrunk lattice $\mathcal{L}^\tsf{pr}$ (a green cell is identical except the colors of \tsf{AQ}s) and (b) red and green cells in the primal blue shrunk lattice $\mathcal{L}^\tsf{pb}$.
    	\tsf{p}$t$s and \tsf{d}$t$s indicate primal and dual layers, respectively.
    	Some qubits on the last layer are not displayed.
    	All the \tsf{pcAQ}s are vertices of $\mathcal{L}^\tsf{pc}$.
    	Each spacelike (or timelike) edge, visualized as red or blue solid lines, connects two adjacent vertices in a layer (or different layers) and corresponds to a \tsf{pcL} (or \tsf{dcAQ}).
    	Faces and cells are defined naturally with the edges.
	}
	\label{fig:CCCS_shrunk_lattices}
\end{figure}


The primal \tsf{c}-colored shrunk lattice $\mathcal{L}^\tsf{pc}$ is a 3D lattice containing every \tsf{pcAQ} as a vertex.
Note that the vertices are only in primal layers.
There are two types of edges connecting them: ``spacelike'' and ``timelike'' edges.
Each spacelike edge corresponds to a \tsf{pcL} and connects two vertices in a layer.
Each timelike edge connects two vertices adjacent along the time axis and contains a \tsf{dcAQ} between them.
Faces and cells are then naturally defined by the vertices and edges.
Cells in each primal shrunk lattice are visualized in Fig.~\ref{fig:CCCS_shrunk_lattices} for 4-8-8 CCCSs.
Note that each primal layer in $\mathcal{L}^\tsf{pc}$ is identical with the \tsf{c}-colored shrunk lattice of the 2D color code on which the CCCS is based.


\begin{table}[b!]
    \caption{
        Qubits $Q(b)$ corresponding to each element (vertex, edge, face, or cell) $b$ in $\mathcal{L}^\tsf{pc}$.
        The results for $\mathcal{L}^\tsf{dc}$ can be obtained by changing each \tsf{p} or \tsf{d}.
    }
    \label{table:qubits_corr_to_each_element_in_shrunk_lattices}
    \centering
    \begin{ruledtabular}
    \begin{tabular}{ccc}
        \multicolumn{2}{c}{Element $b$ in $\mathcal{L}^\tsf{pc}$} & Qubits $Q(b)$ \\ \hline
        \multicolumn{2}{c}{Vertex ($\in \mathcal{B}^\tsf{c}_0$)} & \tsf{pcAQ} \\
        \multirow{2}{*}{Edge ($\in \mathcal{B}^\tsf{c}_1$)} & Timelike & \tsf{dcAQ} \\
                                         & Spacelike      & \tsf{dcL} (two \tsf{dCQ}s)                          \\
        \multirow{2}{*}{Face ($\in \mathcal{B}^\tsf{c}_2$)}       & Timelike       & \tsf{pcL} (two \tsf{pCQ}s)                           \\ 
                                         & Spacelike      & \tsf{pc}$'$\tsf{AQ} ($\tsf{c}' \neq \tsf{c}$) \\
        \multicolumn{2}{c}{Cell ($\in \mathcal{B}^\tsf{c}_3$)}                    & \tsf{dc}$'$\tsf{AQ} ($\tsf{c}' \neq \tsf{c}$)
    \end{tabular}
    \end{ruledtabular}
\end{table}

\begin{figure*}[t!]
	\centering
	\includegraphics[width=\textwidth]{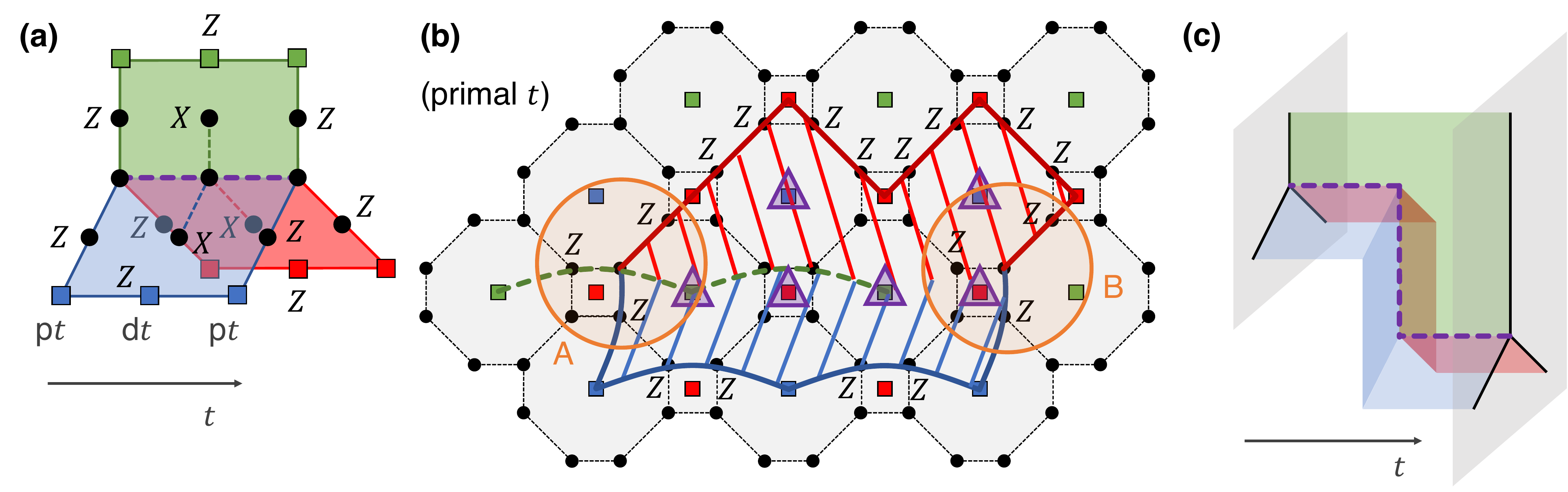}
	\caption{
	    (a) Timelike joint of primal correlation surfaces (CSs) originated from a J-type SG.
	    The $X$ or $Z$ operators on the qubits indicate the support of the resulting \tsf{CS}.
	    A series of \tsf{CQ}s along which the three faces meet is marked as a purple dashed line.
	    (b) Example of the construction of a spacelike joint of three primal \tsf{CS}s.
	    A primal layer of an 4-8-8 CCCS is presented.
	    We first assume a timelike \tsf{pg-CS} $S$ ending at the green dashed line.
	    We then expand $S$ by multiplying the A-type SGs around the \tsf{pAQ}s marked with purple triangles.
	    After the expansion, $\supp_X(S)$ contains the marked \tsf{pAQ}s, and $\supp_Z(S)$ contains the \tsf{CQ}s along the red and blue solid lines.
	    The area above (below) the green line can be regarded as a \tsf{pr(b)-CS}, in the sense that it may be expanded by multiplying ordinary \tsf{pr(b)-CS}s.
	    A joint of the three \tsf{CS}s is thus constructed, and $S$ is the corresponding joined \tsf{CS}.
	    The qubits in $\supp_Z S$ inside the area A or B exactly match with the final layer of a timelike joint, thus spacelike and timelike joints may be connected.
	    (c) Example of a general joint, obtained by multiplying a series of timelike and spacelike joints together with ordinary \tsf{CS}s.
	    }
	\label{fig:joints}
\end{figure*}

Each element (vertex, edge, face, or cell) in a shrunk lattice corresponds to an \tsf{AQ} or a link, as presented in Table~\ref{table:qubits_corr_to_each_element_in_shrunk_lattices}.
Here $Q(b)$ for an element $b$ denotes the set of qubits corresponding to $b$.
Note that an element is colorable with the color of \tsf{AQ} or link corresponding to it.
In particular, cells and spacelike faces have colors different from the color of the shrunk lattice, e.g., $\mathcal{L}^\tsf{pr}$ is composed of green and blue cells.


We now regard the shrunk lattices as chain complexes \cite{raussendorf2006fault, raussendorf2007fault, raussendorf2007topological}.
Let $\mathcal{B}_i^\tsf{pc}$ for $i=0$, 1, 2, or 3 be the set of vertices, edges, faces, or cells in $\mathcal{L}^\tsf{pc}$, respectively.
We then consider a vector space $H_i^\tsf{pc}$ generated by $\mathcal{B}_i^\tsf{pc}$ over $\mathbb{Z}_2$.
Each primal shrunk lattice may be regarded as a chain complex: $\mathcal{L}^\tsf{pc} = \qty{ H_3^\tsf{pc}, H_2^\tsf{pc}, H_1^\tsf{pc}, H_0^\tsf{pc} }$.
Each element $h_i \in H_i^\tsf{pc}$ is called an \textit{$i$-chain} and corresponds to a set $B\qty(h_i) \subseteq \mathcal{B}_i^\tsf{c}$ where each $b \in B\qty(h_i)$ has nonzero contribution in $h_i$.
For example, if $f_1$, $f_2$, and $f_3$ are faces in $\mathcal{L}_2^\tsf{pc}$, $h_2 := f_1 + f_2 + f_3$ is a 2-chain in $H_2^\tsf{pc}$ and $B\qty(h_2) = \qty{ f_1, f_2, f_3 }$ holds.
The correspondence is one-to-one, thus we do not distinguish $h_i$ and $B\qty(h_i)$ from now on if it is not confusing.
The chain complex $\mathcal{L}^\tsf{pc}$ has a boundary map $\partial$ which maps $h_i \in H_i^\tsf{pc}$ to $\partial h_i \in H_{i-1}^\tsf{pc}$ corresponding to the geometrical boundary of $h_i$.
Note that $\partial$ is a linear map and satisfies $\partial \circ \partial = 0$.

For an $i$-chain $h_i$ and $P \in \qty{X, Y, Z}$, we define a multi-qubit Pauli operator $P\qty(h_i)$ by
\begin{align*}
    P\qty(h_i) := \prod_{q \in Q(h_i)} P(q),
\end{align*}
where $Q(h_i) := \bigcup_{b_i \in h_i} Q(b_i)$ and $P(q)$ is the tensor product of the $P$ operator on the qubit $q$ and the identity operator on all the other qubits.
We now define correlation surfaces (\tsf{CS}s), essential elements for constructing logical operations through MBQC.

\begin{definition}[\textbf{Correlation surface}]
    For each 2-chain $h_2 \in H_2^\tsf{p(d)c}$, the operator
    \begin{align}
        S_\tsf{CS}\qty(h_2) := X\qty(h_2) Z\qty(\partial h_2).
        \label{eq:correlation_surface_definition}
    \end{align}
    is a \textup{primal (dual) \tsf{c}-colored correlation surface}, referred to as a ``\tsf{p(d)c-CS}.''
\end{definition}

It is straightforward to see that, for a spacelike or timelike face $f$, $S_\tsf{CS}\qty(f)$ is an A- or L-type SG around the \tsf{AQ} or link corresponding to $f$, respectively.
The following theorem relates general 2-chains to stabilizers of the CCCS.

\begin{theorem}[\textbf{\tsf{CS}s as stabilizers}]
    For a 2-chain $h_2$, $S_\tsf{CS}\qty(h_2)$ is a stabilizer if and only if $Q\qty( h_2 ) \cap Q_\mathrm{IN} = \emptyset$, where $Q_\mathrm{IN}$ is the set of input qubits defined in Sec.~\ref{sec:cluster_states_and_MBQC} which are not initialized to the $\ket{+}$ states.
    \label{thr:cs_and_stabilizers}
\end{theorem}

\begin{proof}
(\textit{If part})
Since qubits outside $Q_\mathrm{IN}$ is initialized to the $\ket{+}$ states, there exist the A- or C-type SG around each of them, as discussed in Sec.~\ref{sec:cluster_states_and_MBQC}.
Let $F := \qty{ f \in \mathcal{B}_2^\tsf{pc} \mid Q(f) \cap Q_\mathrm{IN} = \emptyset }$ be a set of faces.
For a face $f \in F$, $S_\tsf{CS}\qty(f)$ is a stabilizer; it is an A- or L-type SG.
For a 2-chain $h_2 \in H_2^\tsf{pc}$ where $Q\qty( h_2 ) \cap Q_\mathrm{IN} = \emptyset$, $h_2$ can be written as a linear summation of elements in $F$: $\exists \qty{ f_i } \subseteq F$, $h_2 = \sum_i f_i$.
Since the map $\partial$ is linear and $P(h)P\qty(h') = P\qty(h + h')$ holds for any Pauli operator $P$, $S_\tsf{CS}\qty(h_2) = X\qty(h_2) Z\qty(\partial h_2) = \prod_i X\qty(f_i) Z\qty(\partial f_i) = \prod_i S_\tsf{CS}\qty(f_i)$, which is a stabilizer.
The proof is analogous for dual 2-chains.

(\textit{Only if part})
Since qubits in $Q_\mathrm{IN}$ are not initialized to the $\ket{+}$ states, the A- and C-type SGs around each of them do not exist.
Therefore, the $X$-support of any stabilizer cannot contain qubits in $Q_\mathrm{IN}$.
\end{proof}

Regarding a primal \tsf{CS} $S := S_\tsf{CS} \qty( h_2 )$, $Q\qty(h_2)$ is called the \textit{interior} of $S$, in which every qubit is primal and in $\supp_X (S)$.
Similarly, $Q\qty(\partial h_2)$ is called the \textit{boundary} of $S$, in which every qubit is dual and in $\supp_Z (S)$.
We say that $S$ is \textit{timelike} (\textit{spacelike}) if $h_2$ is composed of timelike (spacelike) faces only.


\tsf{CS}s discussed above include all A- and L-type SGs, but not J-type SGs in Fig.~\ref{fig:cluster_state_SGs}(d).
Each J-type SG can be regarded as three primal timelike \tsf{CS}s with different colors ``joined'' along a timelike series of \tsf{CQ}s as Fig.~\ref{fig:joints}(a), in the sense that each ``wing'' of a color \tsf{c} may be extended by multiplying ordinary \tsf{pc-CS}s.
Note that the \tsf{CQ}s along the joint are not included in the support.

\begin{figure*}[t!]
	\centering
	\includegraphics[width=\textwidth]{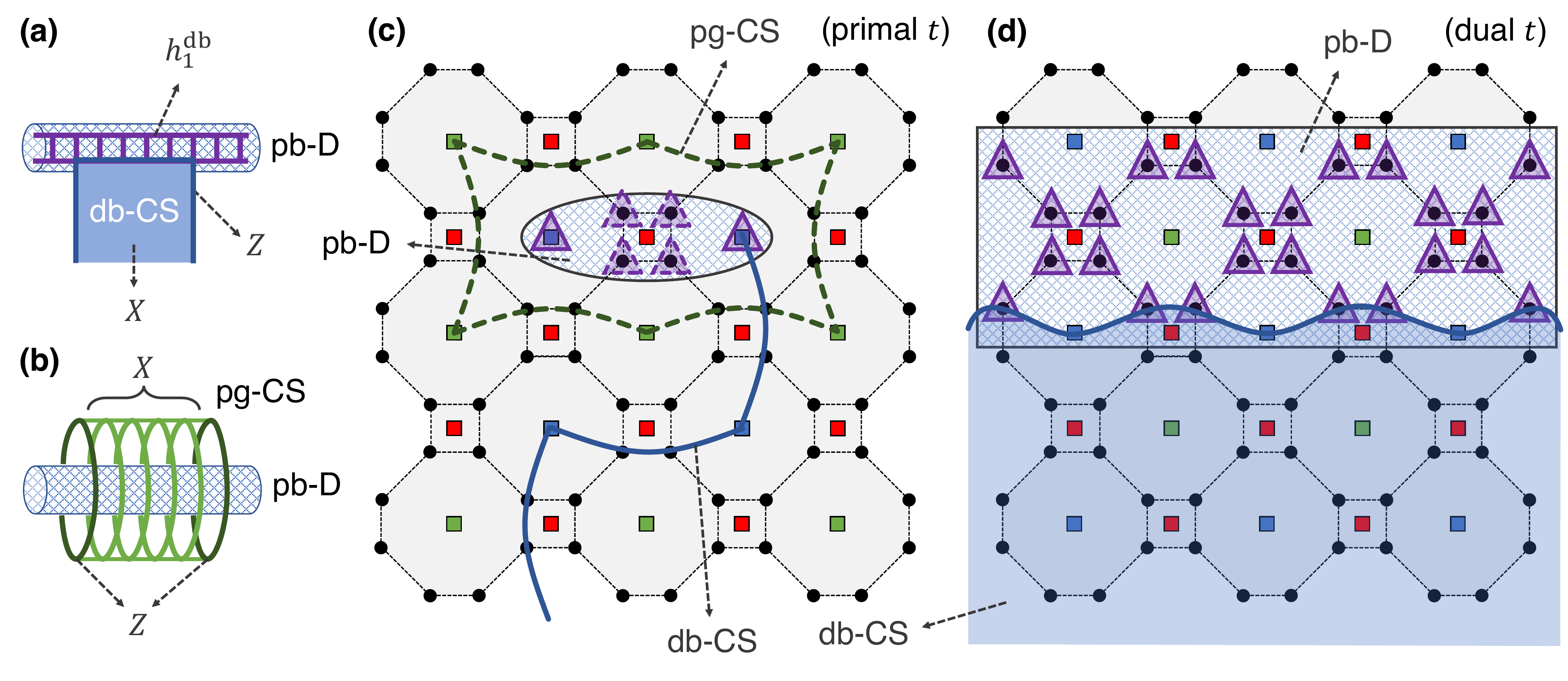}
	\caption{
	    (a) Schematic diagram of a defect (\tsf{pb-D}) and a \tsf{db-CS} $S$ ending at the defect.
	    The cylinder indicates the defect area and the purple lines indicate the defect, a 1-chain $h_1^\tsf{db} \in H_1^\tsf{db}$.
	    Qubits corresponding to $h_1^\tsf{db}$ are measured in the $Z$ basis.
	    (b) Schematic diagram of a \tsf{pg-CS} surrounding a \tsf{pb-D}.
	    (c) A primal layer in a 4-8-8 CCCS penetrated by a timelike \tsf{pb-D}.
	    The cross-section of the defect area is displayed as a checkered pattern.
	    Each purple triangle with a solid (dashed) border indicates a defect \tsf{pbAQ} (\tsf{pCQ}) in the layer (next layer) measured in the $Z$ basis.
	    The cross-sections of a timelike \tsf{db-CS} ending at the defect and a timelike \tsf{pg-CS} surrounding it are presented as a blue solid line and a green dashed line, respectively.
	    The solid (or dashed) lines indicate faces bisected by the layer (or ending at the layer).
	    That is, the corresponding qubits are on the layer (or an adjacent layer).
	    (d) A dual layer in a 4-8-8 CCCS containing a spacelike \tsf{pb-D}.
	    The defect area is displayed as a checkered pattern.
	    A \tsf{db-CS} ending at the defect is visualized as a blue surface, where the blue line corresponds to its boundary.
    }
	\label{fig:defect_structure}
\end{figure*}

A question arising naturally may be about ``spacelike'' joints, and those are also possible as presented in Fig.~\ref{fig:joints}(b).
A timelike \tsf{pc-CS} and two spacelike primal \tsf{CS}s with the other two colors may be joined along a spacelike series of \tsf{pcL}s.
Such a joint can be obtained by multiplying several A-type SGs along a spacelike boundary of the timelike \tsf{CS}.
Note that the ends of spacelike and timelike joints may fit perfectly with each other, in the sense that all the $Z$ operators on the joint cancel out when multiplying them.

A general joint of \tsf{CS}s with different colors can be obtained as Fig.~\ref{fig:joints}(c) by multiplying several spacelike and timelike joints together with ordinary \tsf{CS}s.
We refer to such a primal \tsf{CS} with a joint as a ``\tsf{pj-CS}.''
For consistency with ordinary \tsf{CS}s, we define the \textit{interior} (\textit{boundary}) of a \tsf{pj-CS} by its $X$($Z$)-support, which is intuitive considering its visualization in Fig.~\ref{fig:joints}.

\section{Measurement-based quantum computation via color-code-based cluster states}
\label{sec:MBQC_via_color_code_based_cluster_states}

In this section, we describe MBQC via CCCSs.
We first introduce defects and define logical qubits using them.
We then describe initialization and measurements of logical qubits and construct elementary logical gates including the identity, \cnot, Hadamard, and phase gates, which together generate the Clifford group.
We lastly present the state injection scheme to prepare an arbitrary logical state and implement the logical $T$ gate.

Each logical initialization, measurement, gate, or state injection process can be regarded as an independent circuit ``block'' implemented by the process presented in Sec.~\ref{sec:cluster_states_and_MBQC}.
In each block, the input logical state is given in the input qubits $Q_\mathrm{IN}$ ($Q_\mathrm{IN} = \emptyset$ for the initialization) and the output logical state is produced in the output qubits $Q_\mathrm{OUT}$ ($Q_\mathrm{OUT} = \emptyset$ for the logical measurements) after the single-qubit Pauli measurements of all the qubits except $Q_\mathrm{OUT}$.
An arbitrary quantum circuit can be formed by connecting multiple blocks in a way that the output qubits of each block are used as the input qubits of the next block.

Without loss of generality, we assume that the single-qubit measurements are performed layer by layer along the simulating time ($t$) axis.
In that case, the output qubits of an initialization, gate, or state injection block are the last several layers of it, called the \textit{output layers}.
On the other hands, it is sufficient that the input qubits of a measurement or gate block contain only the first layer of it, called the \textit{input layer}.
Two subsequent blocks can be connected in a way that the input layer of the second block overlap with the first output layer of the first block.
To see this, let us assume that the output layers of the first block are the layers of $t_0 \leq t \leq t_1$.
We first consider applying all the \cz~gates between qubits of $t_0 \leq t \leq t_1$ again on the post-measurement state of the first block.
Since the measurements of the qubits of $t < t_0$ commute with those \cz~gates, the qubits of $t_0 < t \leq t_1$ simply return to the initial $\ket{+}$ states.
The $t_0$-layer is then used for the input layer of the second block and the \cz~gates in the second block restore the output state of the first block to be used as the input state of the second block.
Of course, the above argument is just a theoretical trick to connect two blocks; it is unnecessary to apply \cz~gates multiple times in a real implementation.

\subsection{Measurement pattern}
\label{subsec:measurement_pattern}

Remark that each qubit except the output qubits is measured in a Pauli basis determined by a predefined measurement pattern.
Such a qubit is included in an area with one of the four types: \textit{vacuum}, \textit{defect}, \textit{Y-plane}, and \textit{injection qubit}.
There may be multiple defects, Y-planes, and injection qubits, and the entire remaining area is the vacuum.
We denote the set of all vacuum (defect) qubits as $V$ ($D$).

Defects are key ingredients for the protocol; all the logical operations completely depend on how to place them.
Y-planes are used in fault-tolerant $Y$-measurements on physical qubits for the logical Hadamard and phase gates.
Lastly, each injection qubit is a special area for state injection and consists of a single qubit.
Qubits in each area are measured as follows:
\begin{align}
    \begin{array}{cc}
        \text{A qubit is measured} \\
        \text{in the basis of}
    \end{array}
    \left\{
        \begin{array}{ll}
            X & \text{if in the vacuum} \\
              & \text{or an injection qubit,} \\
            Z & \text{if in a defect,} \\
            Y & \text{if in a Y-plane.} \\
        \end{array}
    \right. \label{eq:measurement_pattern}
\end{align}
Arranging these elements besides the vacuum properly is the key for implementing logical qubits and gates, which is what we cover in this section.

\subsection{Defects and related correlation surfaces}
\label{subsec:measurement_pattern_and_defects}

Defects are defined as follows and visualized in Fig.~\ref{fig:defect_structure}(a) schematically.

\begin{definition}[\textbf{Defect}]
    Consider a \textup{defect area}, a continuous area stretching along a direction.
    A \textup{primal (dual) \tsf{c}-colored defect}, referred to as a ``\tsf{p(d)c-D},'' is the largest 1-chain $h_1 \in H_1^\tsf{d(p)c}$ contained completely in the defect area. 
    Qubits in $Q\qty(h_1)$ are called \textup{defect qubits} and measured in the $Z$ basis during the measurement step.
\end{definition}

We say that a defect is \textit{timelike} or \textit{spacelike} if the defect area stretches timelikely or spacelikely, respectively.
Figure \ref{fig:defect_structure}(c) and (d) illustrate the explicit structures of timelike and spacelike defects, respectively, in an 4-8-8 CCCS.

After the measurement step, only \tsf{CS}s commuting with the measurement pattern survive.
We say a \tsf{CS} is \textit{compatible} with a set of qubits if it survives after the measurements of the qubits.
If it is compatible with all the qubits except the output qubits, we say that it is a compatible \tsf{CS}.
The following theorem gives the conditions which such \tsf{CS}s satisfy.

\begin{theorem}[\textbf{Compatible \tsf{CS}s}]
\begin{subequations}
Considering only the vacuum and defects, a \tsf{CS} $S$ is compatible with a set $\widetilde{Q}$ of qubits if and only if the followings hold:
\begin{align}
    Q_\mathrm{int}(S) \cap \widetilde{Q} \setminus Q_{\mathrm{OUT}} &\subseteq V \setminus Q_\mathrm{IN}, \label{eq:compatible_CS_X_cond}\\
    Q_\mathrm{bnd}(S) \cap \widetilde{Q} \setminus Q_{\mathrm{OUT}} &\subseteq D, \label{eq:compatible_CS_Z_cond}
\end{align}
where $Q_\mathrm{int(bnd)}(S)$ is the interior (boundary) of $S$ and $Q_\mathrm{IN(OUT)}$ is the set of input (output) qubits.
\label{eq:compatible_CS_conds}
\end{subequations}
\label{thr:compatible_CSs}
\end{theorem}

We particularly want to emphasize that a compatible \tsf{CS} cannot end in the vacuum qubits.
Note that $Q_\mathrm{IN}$ is excluded in the right-hand side (RHS) of Eq.~\eqref{eq:compatible_CS_X_cond} due to Theorem \ref{thr:cs_and_stabilizers}.


\begin{table}[b!]
    \caption{
        Allowed positional relations between a primal defect $d$ and a compatible \tsf{CS}.
        The relations for dual defects are analogous.
    }
    \label{table:CS_defect_relationship}
    \centering
    \begin{ruledtabular}
    \begin{tabular}{ccc}
        \multicolumn{3}{c}{With a \tsf{pc-D} $d$, a \tsf{\textbf{xy}-CS} $S$ ...} \\
        \backslashbox{\textbf{\tsf{x}}}{\textbf{\tsf{y}}}     & \tsf{c}        &  $\tsf{c}' (\neq \tsf{c})$ \\ \hline
        \tsf{p} & \begin{tabular}[c]{@{}c@{}} can overlap with $d$ \\ only if $S$ is spacelike \end{tabular} & cannot overlap with $d$ \\
        \tsf{d} & can end at $d$        & cannot end at $d$
    \end{tabular}
    \end{ruledtabular}
\end{table}

Table~\ref{table:CS_defect_relationship} shows allowed positional relations between a \tsf{pc-D} $d$ and a compatible \tsf{CS} with each primality and color, derived from Theorem \ref{thr:compatible_CSs} and Table~\ref{table:qubits_corr_to_each_element_in_shrunk_lattices}.
Remark that $d$ is composed of \tsf{pcAQ}s and \tsf{pCQ}s.
A \tsf{pc}$'$\tsf{-CS} has primal interior qubits which should be measured in the $X$ basis for the \tsf{CS} to be compatible, thus can overlap with $d$ unless they share common qubits, which is possible only if $\tsf{c}' = \tsf{c}$ and the overlapped region of the \tsf{CS} is spacelike \footnote{
    If this is the case, the interior qubits of the \tsf{CS} correspond to spacelike faces in $\mathcal{L}^\tsf{pc}$, which are $\tsf{pc}''\tsf{AQ}$s ($\tsf{c}'' \neq \tsf{c}$) according to Table~\ref{table:qubits_corr_to_each_element_in_shrunk_lattices}. 
    They are surely not be in the \tsf{pc-D}. 
    Otherwise, the \tsf{pc}$'$\tsf{-CS} and defect share at least one qubit if they overlap.
}.
Since its boundary qubits are dual, it cannot end at $d$.
A \tsf{dc}$'$\tsf{-CS} has primal boundary qubits, thus can end at $d$ if the boundary qubits are in $d$, which is possible if $\tsf{c}' = \tsf{c}$ \footnote{
    The boundary qubits of a \tsf{dc-CS} correspond to edges in $\mathcal{L}^\tsf{dc}$, which are \tsf{pcAQ}s or \tsf{pcL}s.
}.
Since its interior qubits are dual, it can freely pass $d$.

We mainly concern two types of \tsf{CS}s with respect to a \tsf{pc-D}: \tsf{pc-CS}s surrounding the defect and \tsf{dc-CS}s ending at it, as shown schematically in Fig.~\ref{fig:defect_structure}(a) and (b) and explicitly in Fig.~\ref{fig:defect_structure}(c) and (d).
Each of such \tsf{CS}s is compatible with all the qubits except the boundary qubits in the two ends about the direction of the defect.

\subsection{Defining a logical qubit}
\label{subsec:logical_qubit}

We first define \textit{connected 1-chains} as follows.

\begin{definition}[\textbf{Connected 1-chain}]
    A 1-chain $h_1$ is \textup{connected} if and only if it satisfies $\qty| \partial h_1 | \leq 2$.
    It is regarded to be \textup{closed} if $\qty| \partial h_1 | = 0$ and \textup{open} otherwise.
\end{definition}

\begin{figure*}[t!]
	\centering
	\includegraphics[width=\textwidth]{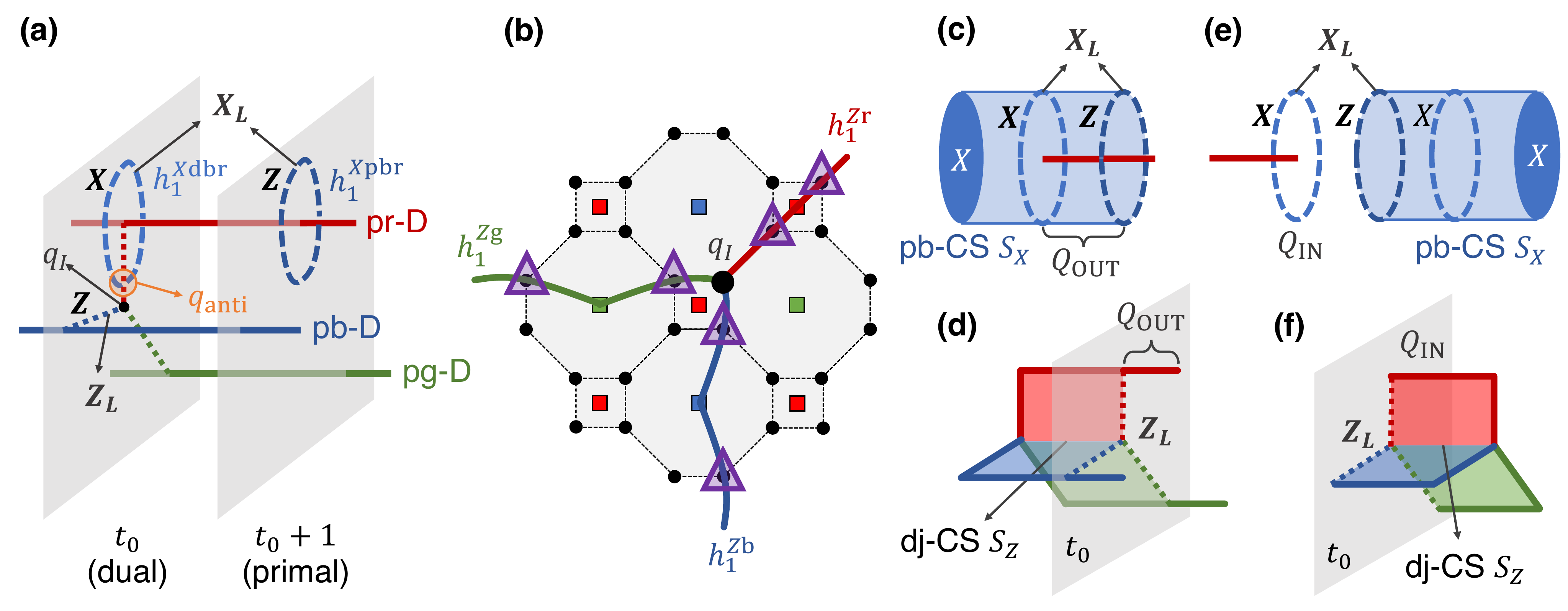}
	\caption{
	    Definition of a primal logical qubit and its initialization and measurement.
	    (a) Schematic diagram of a primal logical qubit composed of three parallel primal timelike defects with different colors.
	    Blue dashed lines indicate 1-chains $h_1^{X\tsf{dbr}}$ and $h_1^{X\tsf{pbr}}$, which constitute $\supp_X\qty(X_L)$ and $\supp_Z\qty(X_L)$, respectively.
	    Red, green, and blue dotted lines indicate 1-chains $h_1^{Z\tsf{r}}$, $h_1^{Z\tsf{g}}$, and $h_1^{Z\tsf{b}}$, respectively, which constitute $\supp_Z\qty(Z_L)$ except the \tsf{pCQ} $q_I$ at which they end.
	    $\supp_X\qty(X_L)$ and $\supp_Z\qty(Z_L)$ meet at a \tsf{pCQ} $q_\mathrm{anti}$, thus they anticommute with each other.
	    (b) Structure of $Z_L$ near $q_I$ in a 4-8-8 CCCS.
	    Colored lines are $h_1^{Z\tsf{r}}$, $h_1^{Z\tsf{g}}$, and $h_1^{Z\tsf{b}}$, respectively.
	    Purple triangles indicate $\supp\qty(Z_L)$.
	    (c) $X_L$- and (d) $Z_L$-initialization.
	    A logical qubit prepared in the output layers $Q_\mathrm{OUT}$ ($t_0$- and $\qty(t_0+1)$-layer).
	    For the $X_L$-initialization, the defects are made to start from the $t_0$-layer.
	    For the $Z_L$-initialization, they are extended to meet at a point before the layer-$t_0$.
	    $X_L$ ($Z_L$) is then a part of a \tsf{pb-CS} $S_X$ (\tsf{dj-CS} $S_Z$) which is a stabilizer.
	    After the measurement step, the logical qubit in $Q_\mathrm{OUT}$ is initialized to $\ket{\pm_L}$ ($\ket{0_L}$ or $\ket{1_L}$), depending on the measurement result of $X_L S_X$ ($Z_L S_Z$).
	    (e) $X_L$- and (f) $Z_L$-measurement of a logical qubit inserted through the input layer ($t_0$-layer).
    	Each of them is done by reversing the corresponding initialization process.
    	There then exists a \tsf{pb-CS} $S_X$ (\tsf{dj-CS} $S_Z$) which is a stabilizer, such that the measurement result of $S_X X_L$ ($S_Z Z_L$) determine the $X_L$($Z_L$)-measurement result.
	}
	\label{fig:logical_qubit}
\end{figure*}

To define a logical qubit, we consider three parallel timelike defects with different colors passing through the $t_0$- and $\qty(t_0+1)$-layer for a given integer $t_0$, as visualized schematically in Fig.~\ref{fig:logical_qubit}(a).
The constructed logical qubit is primal (dual) if the defects are primal (dual) and $t_0$ is odd (even).

We define a logical qubit by specifying the logical-$X$ ($X_L$) and logical-$Z$ ($Z_L$) operators.
To define $X_L$, we consider two closed connected spacelike 1-chains $h_1^{X\tsf{dcc}'} \in H^\tsf{dc}_1$ and $h_1^{X\tsf{pcc}'} \in H^\tsf{pc}_1$ for a given pair of different colors $\qty( \tsf{c}, \tsf{c}' )$.
$h_1^{X\tsf{dcc}'}$ is located in the $t_0$-layer and surrounding the $\tsf{pc}'\tsf{-D}$.
$h_1^{X\tsf{pcc}'}$ is defined by parallelly moving $h_1^{X\tsf{dcc'}}$ one unit positively along the time axis.
An example of $X_L$ is shown in Fig.~\ref{fig:logical_qubit}(a) for the case of $\qty( \tsf{c}, \tsf{c}' ) = \qty( \tsf{b}, \tsf{r} )$.
Note that the two 1-chains consist of \tsf{pcL}s and \tsf{dcL}s, respectively.
We then define
\begin{align}
    X_L := F_{X}^{\tsf{cc}'}\qty(t_0) := X\qty(h_1^{X\tsf{dcc}'}) Z\qty(h_1^{X\tsf{pcc}'}).
    \label{eq:logical_X_def}
\end{align}
Note that $\supp_Z\qty(X_L) = Q\qty(h_1^{X\tsf{pcc'}})$ may be in the boundary of a $\tsf{pc-CS}$ since the boundary is a 1-chain in $H^\tsf{pc}_1$ as well.
The colors \tsf{c} and $\tsf{c}'$ can be any pair of different colors, and they are proven to be equivalent in Sec.~\ref{subsubsec:cnot_gate}.

For the $Z_L$ operator, we consider an open connected spacelike 1-chain $h_1^{Z\tsf{c}} \in H^\tsf{dc}_1$ for each color $\tsf{c}$, which is located in the $t_0$-layer and connects the \tsf{pc-D} and a common \tsf{pCQ} $q_I$, as shown in Fig.~\ref{fig:logical_qubit}(a) and (b).
Note that $h_1^{Z\tsf{c}}$ is composed of \tsf{pcL}s.
We define
\begin{align}
    Z_L &:= F_Z \qty(t_0) \nonumber \\
    &:= Z\qty(h_1^{Z\tsf{r}}) Z\qty(h_1^{Z\tsf{g}}) Z\qty(h_1^{Z\tsf{b}}) Z\qty(q_I).
    \label{eq:logical_Z_def}
\end{align}
Note that $q_I$ is out of $\supp\qty(Z_L)$.
It is worth noticing that $\supp\qty(Z_L)$ may be in the boundary of a \tsf{dj-CS}, which is verifiable by comparing $\supp\qty(Z_L)$ and the structure of a timelike joint of \tsf{CS}s shown in Fig.~\ref{fig:joints}(c).

$X_L$ and $Z_L$ defined above anticommute with each other, considering that $\supp_Z\qty(Z_L)$ and $\supp_X\qty(X_L)$ meet at a \tsf{pCQ} $q_\mathrm{anti}$ in Fig.~\ref{fig:logical_qubit}(a). A dual logical qubit is defined analogously, but now the logical operators are defined oppositely; $\supp\qty(Z_L)$ surrounds a defect and $\supp\qty(X_L)$ ends at each defect.




\subsection{Initialization and measurement of a logical qubit}
\label{subsec:initialization}


We first describe initializing a primal logical qubit to an eigenstate of $X_L$ or $Z_L$. 
Initialization of a dual logical qubit can be analogously done. 
As mentioned at the beginning of this section, in each block of logical initialization, there is no input layer and the initialized state is prepared in the output layers $Q_\mathrm{OUT}$ ($t_0$- and $\qty(t_0+1)$-layer) after the measurement step.

The $X_L$-initialization of a primal logical qubit is done by making the defects start from the $t_0$-layer.
$X_L$ given in Eq.~\eqref{eq:logical_X_def} is then a part of a ``cup-shaped'' \tsf{pc-CS} $S_X$ as shown in Fig.~\ref{fig:logical_qubit}(c).
Since $X_L S_X$ has the support out of the output qubits and commutes with each single-qubit measurement in the measurement step, the post-measurement state is an eigenstate of $X_L S_X$.
$S_X$ is a stabilizer both before and after the measurement step due to Theorems \ref{thr:cs_and_stabilizers} and \ref{thr:compatible_CSs}.
Therefore, the post-measurement state is also an eigenstate of $X_L$ and the eigenvalue is determined by the measurement result of $X_L S_X$.

The $Z_L$-initialization of a primal logical qubit is done by extending the defects to meet at a qubit before the $t_0$-layer, as shown in Fig.~\ref{fig:logical_qubit}(d).
$Z_L$ given in Eq.~\eqref{eq:logical_Z_def} is then a part of a \tsf{dj-CS} $S_Z$ which is a stabilizer.
From an analogous argument, the post-measurement state is an eigenstate of $Z_L$ and the eigenvalue is determined by the measurement result of $Z_L S_Z$.

The $X_L$- or $Z_L$-measurement is done by reversing the time order from the corresponding initialization process, as shown in Fig.~\ref{fig:logical_qubit}(e) and (f).
This time, $Q_\mathrm{IN}$ is the $t_0$-layer and $Q_\mathrm{OUT}$ is empty.
Regarding the $X_L$-measurement, there exists a \tsf{pb-CS} $S_X$ which is a stabilizer before the measurement step such that $X_L S_X$ commutes with each single-qubit measurement in the measurement step.
Therefore, $X_L$ is \textit{equivalent} to $X_L' := X_L S_X$; namely, $\expval{X_L}{\psi} = \expval{X_L'}{\psi}$ holds for every stabilized state $\ket{\psi}$ before the measurement step, thus redefining $X_L$ to $X_L'$ does not change the logical state encoded in $\ket{\psi}$.
The measurement result of $X_L'$ can be directly obtained from the results of the measurement step.
The $Z_L$-measurement process can be verified analogously.

\subsection{Elementary logical gates}
\label{subsec:elementary_gates}


\subsubsection{Identity gate}
\label{subsubsec:identity_gate}

\begin{figure}[t!]
	\centering
	\includegraphics[width=\columnwidth]{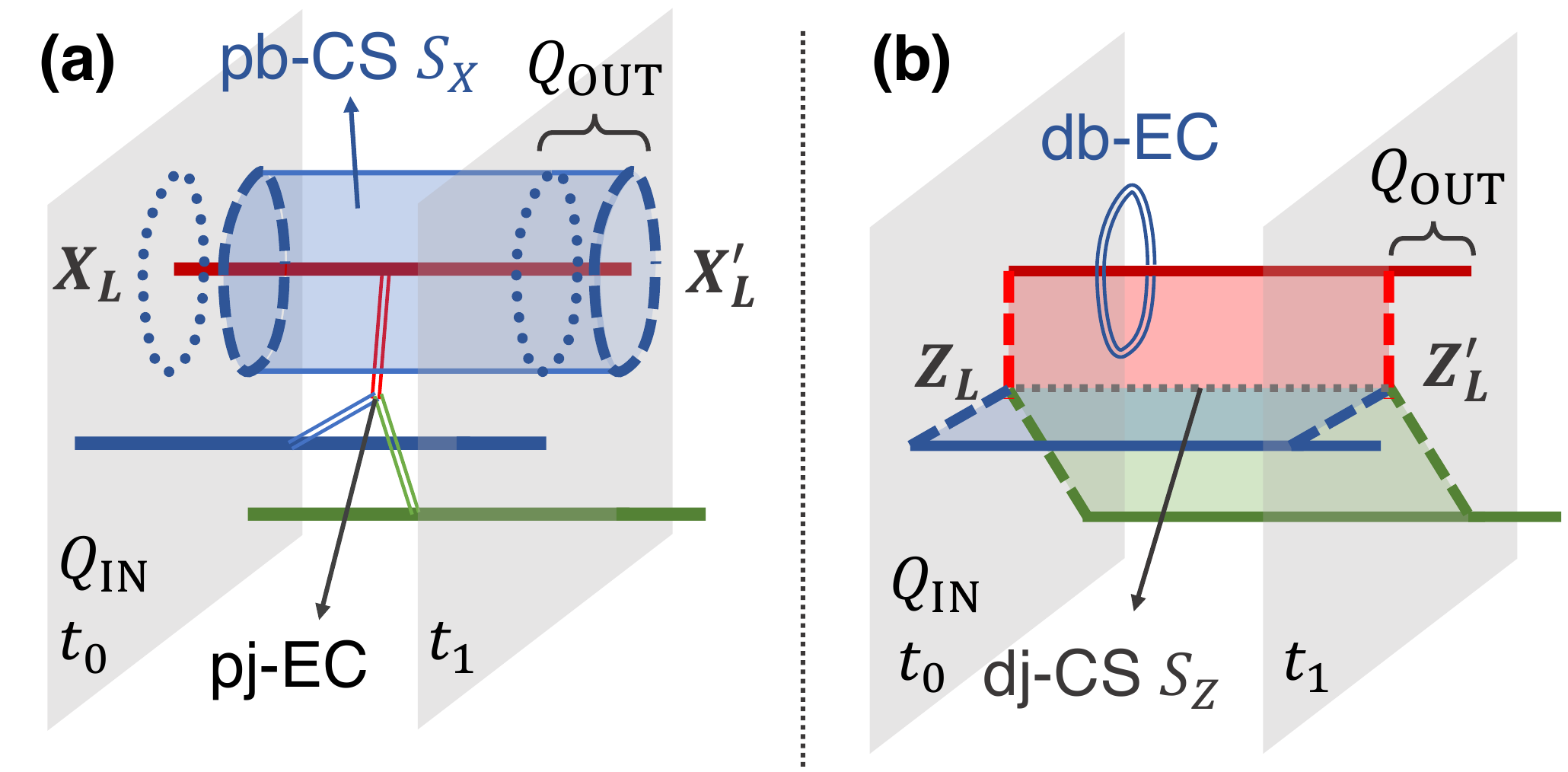}
	\caption{
    	Logical identity gate of a primal logical qubit between the input layer $Q_\mathrm{IN}$ ($t_0$-layer) and the output layers $Q_\mathrm{OUT}$ ($t_1$- and $\qty(t_1+1)$-layer).
    	The gate is constructed by extending the defects from $Q_\mathrm{IN}$ to $Q_\mathrm{OUT}$.
	    The logical-$X$ operator in $Q_\mathrm{IN}$ ($Q_\mathrm{OUT}$) is $X_L$ ($X'_L$), and $Z_L$ and $Z'_L$ are defined similarly.
    	(a) $X_L$ is transformed into $X'_L$ via a \tsf{pb-CS} $S_X$ surrounding the red defect, and (b) $Z_L$ is transformed into $Z'_L$ via a \tsf{dj-CS} $S_Z$ ending at the three defects.
    	Double lines indicate error chains causing logical errors covered in Sec.~\ref{subsec:error_correction_vacuum}.
	}
	\label{fig:identity_gate}
\end{figure}

The identity gate of a primal logical qubit is constructed just by extending the defects along the time axis between $Q_\mathrm{IN}$ ($t_0$-layer) and $Q_\mathrm{OUT}$ ($t_1$- and $\qty(t_1+1)$-layer) as shown in Fig.~\ref{fig:identity_gate}.
Let $X_L$ and $X'_L$ be the logical-$X$ operators of the input and output logical qubits, respectively: $X_{L} := F_X^{\tsf{br}}(t_0)$ and $X_{L}' := F_X^{\tsf{br}}(t_1)$, where $F_X^{\tsf{br}} (\cdot)$ is given in Eq.~\eqref{eq:logical_X_def}.
We consider a \tsf{pb-CS} $S_X$ which surrounds the \tsf{pr-D} and ends at $\supp_Z\qty(X_L)$ and $\supp_Z\qty(X_L')$, as shown in Fig.~\ref{fig:identity_gate}(a).
Since $S_X$ is a stabilizer, $X_L$ is equivalent to
\begin{align}
    \widetilde{X}_L := S_X X_L = \qty( \bigotimes_{q \in V_X} X(q)) X_{L}',
    \label{eq:identity_gate_XL_relation}
\end{align}
where $V_X := \supp\qty(S_X X_L X_L') \subset V\setminus Q_\mathrm{OUT}$. 
After the measurements of the qubits of $t < t_1$, $\widetilde{X}_L$ is \textit{transformed} into 
\begin{align*}
    \qty( \prod_{q \in V_X} x_q ) X'_L := x_X X'_L,
\end{align*}
where $x_q$ ($z_q$) is the $X$($Z$)-measurement result of the qubit $q$. In other words, 
\begin{align}
    \expval{\widetilde{X}_L}{\psi} = \expval{x_X X_L'}{\psi'}
    \label{eq:identical_exp_val}
\end{align}
holds, where $\ket{\psi}$ and $\ket{\psi'}$ are the states before and after the measurements (see Appendix \ref{app:trans_log_op} for the proof).

We do a similar thing on the $Z_L$ operators.
Denoting those of the input and output logical qubits as $Z_L$ and $Z_L'$, respectively, we consider a \tsf{dj-CS} $S_Z$ ending at $\supp\qty(Z_L)$, $\supp\qty(Z'_L)$, and the defects, as Fig.~\ref{fig:identity_gate}(b).
$Z_L$ is then equivalent to
\begin{align*}
    \widetilde{Z}_L := S_Z Z_L = \qty( \bigotimes_{q \in V_Z} X(q)) \qty( \bigotimes_{q \in D_Z} Z(q)) Z_L',
\end{align*}
where $V_Z := \supp_X \qty(S_Z Z_L Z_L') \subset V \setminus Q_\mathrm{OUT}$ and $D_Z := \supp_Z \qty(S_Z Z_L Z_L') \subseteq D \setminus Q_\mathrm{OUT}$.
After the measurements of the qubits of $t < t_1$, $\widetilde{Z}_L$ transforms into $x_Z z_Z Z_L'$ where $x_Z := \prod_{q \in V_Z} x_q$ and $z_Z := \prod_{q \in D_Z} z_q$.

The transformations of the logical operators are summarized as
\begin{align}
    X_L \rightarrow x_X X_L', \qquad Z_L \rightarrow x_Z z_Z Z_L'.
    \label{eq:identity_gate_transformation}
\end{align}
More explicitly, they are written as
\begin{align*}
    \expval{X_L}{\psi} &= \expval{x_X X_L'}{\psi'}, \\
    \expval{Z_L}{\psi} &= \expval{x_Z z_Z Z_L'}{\psi'},
\end{align*}
where $\ket{\psi}$ and $\ket{\psi'}$ are the states before and after the measurement step, respectively.
Therefore, the input logical state $\ket{\psi_L}$ encoded in $\ket{\psi}$ with the logical Pauli operators $\qty{X_L, Z_L}$ is transformed into
\begin{align*}
    \ket{\psi_L'} = X^{(1 - x_Z z_Z)/2} Z^{(1 - x_X)/2} \ket{\psi_L}
\end{align*}
encoded in $\ket{\psi'}$ with the logical Pauli operators $\qty{X_L', Z_L'}$.
This transformation corresponds to the identity gate up to some byproduct operators determined by the measurement results.
The byproduct operators can be handled by a software to be delayed to the end of the entire circuit and finally merged with the logical measurements \cite{fowler2012surface}.

The above arguments show the basic ideas for implementing logical gates.
Regarding $n$ logical qubits, let $P_{Li}$ for each $P \in \qty{X, Z}$ and integer $i \leq n$ denote the logical-$P$ operator of the $i$th logical qubit.
To construct a general logical gate $U$ for $n$ logical qubits, one should find a configuration of defects (and Y-planes for some gates) where a \tsf{CS} $S_{Pi}$ exists for each $P_{Li}$ satisfying the following conditions:

\begin{condition}
    $S_{Pi}$ should connect $P_{Li}$ of the input logical qubits and $U P_{Li} U^\dagger$ of the output logical qubits.
    $X_L$ ($Z_L$) of a logical qubit can be connected with primal (dual) \tsf{CS}s.
    \label{cond:condition_1}
\end{condition}

\begin{condition}
    $S_{Pi}$ should be compatible with all the qubits except the output qubits and $\supp\qty(P_{Li})$; it satisfies the relationships shown in Table~\ref{table:CS_defect_relationship} in that region.
    \label{cond:condition_2}
\end{condition}

If such \tsf{CS}s exist, the configuration implements the desired logical gate with some byproduct operators obtained from the measurement results.


\subsubsection{\cnot~and primality-switching gates}
\label{subsubsec:cnot_gate}

\begin{figure}[t!]
	\centering
	\includegraphics[width=\columnwidth]{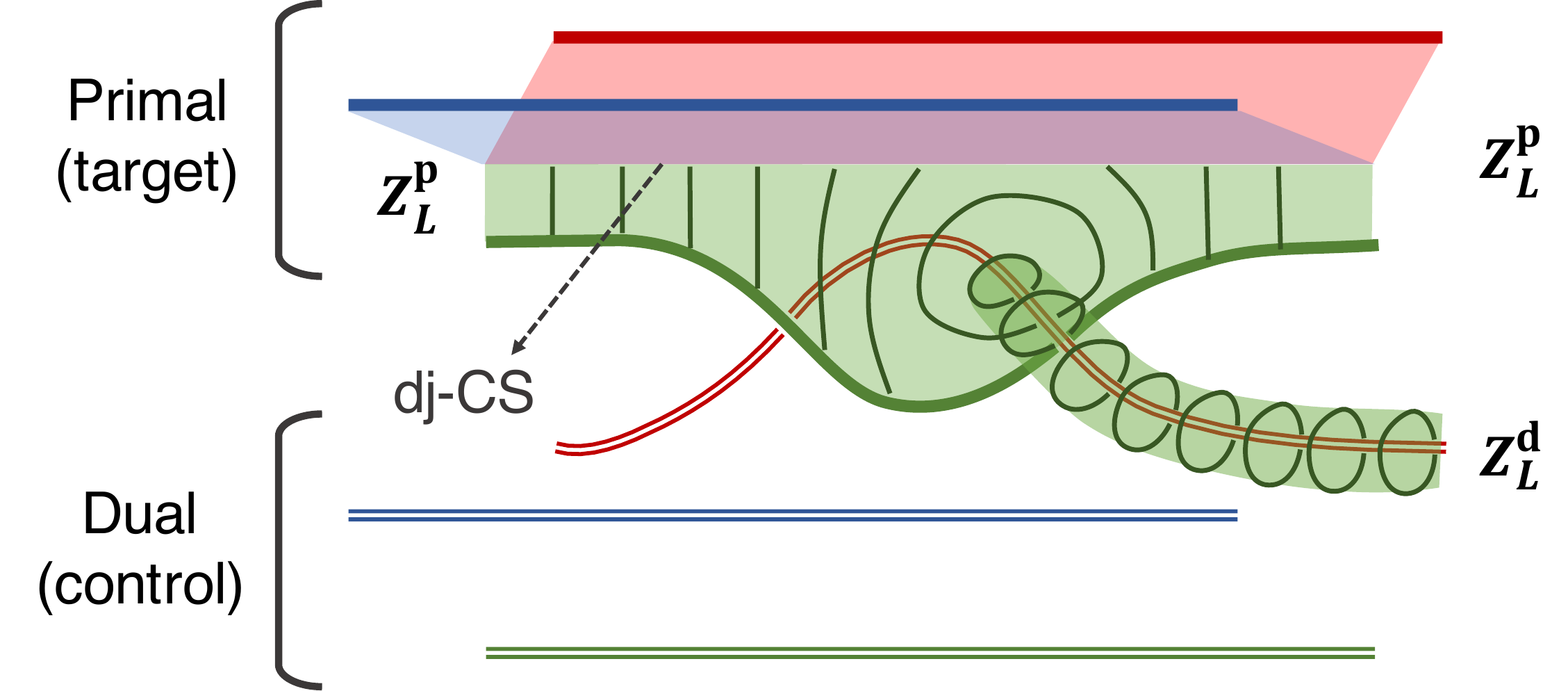}
	\caption{
        Construction of the \cnot~gate between a primal logical qubit (target) and a dual one (control).
        Each colored single (double) line indicates the primal (dual) defect of the corresponding color.
        $Z_L^\tsf{p} \otimes I_L^\tsf{d}$ is transformed into $Z_L^\tsf{p} \otimes Z_L^\tsf{d}$ via the presented \tsf{dj-CS}.
	}
	\label{fig:cnot_gate}
\end{figure}

We first consider the logical \cnot~gate between a primal logical qubit (target) and a dual one (control).
Figure~\ref{fig:cnot_gate} illustrates the defect configuration, where the \tsf{pg-D} of the primal logical qubit and the \tsf{dr-D} of the dual one are twisted one round with each other, which is commonly called \textit{defect braiding}.
The logical Pauli operators are transformed as
\begin{align}
    X_{L}^\tsf{p} I_{L}^\tsf{d} &\rightarrow X_{L}^\tsf{p} I_{L}^\tsf{d}, \qquad
    I_{L}^\tsf{p} X_{L}^\tsf{d} \rightarrow X_{L}^\tsf{p} X_{L}^\tsf{d}, \nonumber \\
    Z_{L}^\tsf{p} I_{L}^\tsf{d} &\rightarrow Z_{L}^\tsf{p} Z_{L}^\tsf{d}, \qquad
    I_{L}^\tsf{p} Z_{L}^\tsf{d} \rightarrow I_{L}^\tsf{p} Z_{L}^\tsf{d}, \label{eq:cnot_gate_transformation}
\end{align}
where the tensor product symbols and the sign terms such as $x_X$, $x_Z$, and $z_Z$ in Eq.~\eqref{eq:identity_gate_transformation} are omitted, and each superscript \tsf{p} or \tsf{d} indicates the primality of the logical qubit.
The above transformation is exactly the Heisenberg picture of the \cnot~gate where the primal logical qubit is the target.

We need to find \tsf{CS}s satisfying two Conditions presented in Sec.~\ref{subsubsec:identity_gate} to verify the transformations in Eq.~\eqref{eq:cnot_gate_transformation}.
A dual \tsf{CS} for the transformation of $Z_L^\tsf{p} \otimes I_L^\tsf{d}$ is presented schematically in Fig.~\ref{fig:cnot_gate}.
Note that the ``tunnel'' of the \tsf{CS} along the \tsf{dr-D} must be formed since the \tsf{dr-D} cannot overlap with a \tsf{dg-CS} (see Table~\ref{table:CS_defect_relationship}).
A \tsf{CS} for $I_L^\tsf{p} \otimes X_L^\tsf{d}$ can be constructed analogously; now, a tunnel of a \tsf{pr-CS} is made along the \tsf{pg-D}.
The other two transformations are straightforward.


\begin{figure}[t!]
    \centering
    \includegraphics[width=\columnwidth]{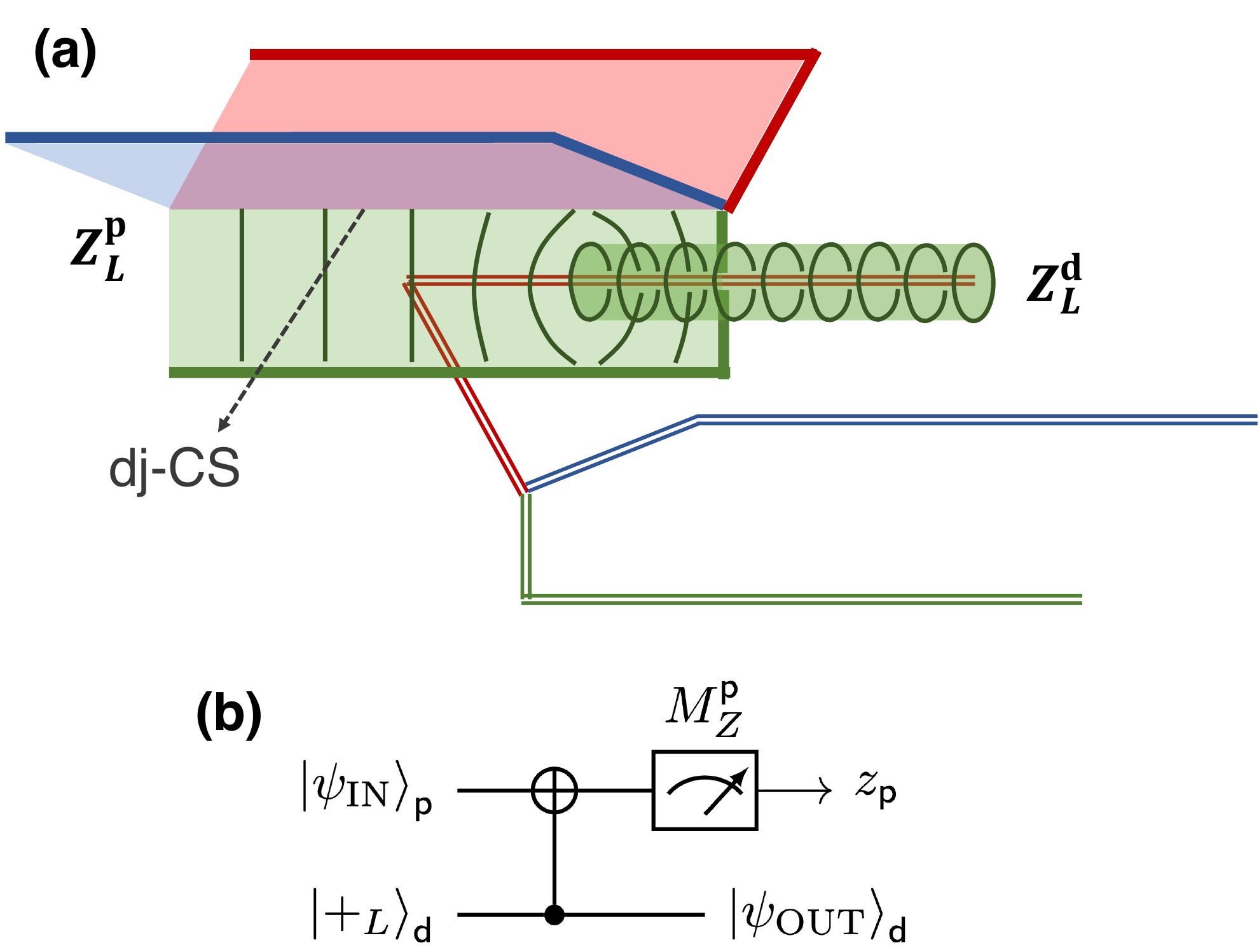}
    \caption{
        (a) Construction of the primality-switching gate changing a primal logical qubit to a dual one.
        $Z_L^\tsf{p}$ is transformed into $Z_L^\tsf{d}$ via the presented \tsf{dj-CS}.
        (b) Circuit equivalent to the primality-switching gate.
        $M_Z^\tsf{p}$ is the $Z_L$-measurement on the primal qubit, and the result is $z_\tsf{p}$.
    }
    \label{fig:PS_identity_gate}
\end{figure}

Exploiting the \cnot~gate discussed above, it is possible to make the \textit{primality-switching} gate which changes a primal logical qubit to a dual one, by ``closing'' the input part of the dual one and the output part of the primal one, as shown in Fig. \ref{fig:PS_identity_gate}(a).
Remark that these closures indicate the $Z_L$-measurement of the primal one and the $X_L$-initialization of the dual one.
The modified configuration is thus equivalent to the circuit in Fig.~\ref{fig:PS_identity_gate}(b) up to byproduct operators, which implements the identity or $X_L$ gate while changing the primality.
Alternatively, this result is directly obtainable by finding appropriate \tsf{CS}s; for example, the \tsf{dj-CS} in Fig.~\ref{fig:PS_identity_gate}(a) verify the transformation of $Z_L^\tsf{p}$ to $Z_L^\tsf{d}$.
The primality-switching gate from a dual logical qubit to a primal one can be made in a similar manner.

The primality-switching gate enables the \cnot~gate between logical qubits with arbitrary primalities.
Regardless of the primalities of the input logical qubits, one can switch them to primal (target) or dual (control), and apply the \cnot~gate in Fig.~\ref{fig:cnot_gate}.

Note that the equivalence between the different definitions of the $X_L$ operator, related to the choice of the color pair $(\tsf{c}, \tsf{c}')$ in Eq.~\eqref{eq:logical_X_def}, can be proven with the primality-switching gate.
We consider a chain of two primality-switching gates: primal $\rightarrow$ dual $\rightarrow$ primal.
No matter how $X_L$ is defined in the first primal logical qubit, it becomes symmetric about the color in the dual one.
We can thus transform it into any definition of $X_L$ in the final primal one.

\subsubsection{Hadamard gate}
\label{subsubsec:hadamard_gate}

To construct the logical Hadamard gate, the logical Pauli operators should be transformed as
\begin{align}
    X_L \rightarrow Z_L, \qquad Z_L \rightarrow X_L.
    \label{eq:hadamard_gate_transformation}
\end{align}
It is simple if the gate is located just after the state injection presented in the Sec.~\ref{subsec:state_injection}: injecting the unencoded state to a dual logical qubit instead of a primal one. 
This method is valid since the definitions of $X_L$ and $Z_L$ are opposite for primal and dual logical qubits.

\begin{figure}[t!]
	\centering
	\includegraphics[width=\columnwidth]{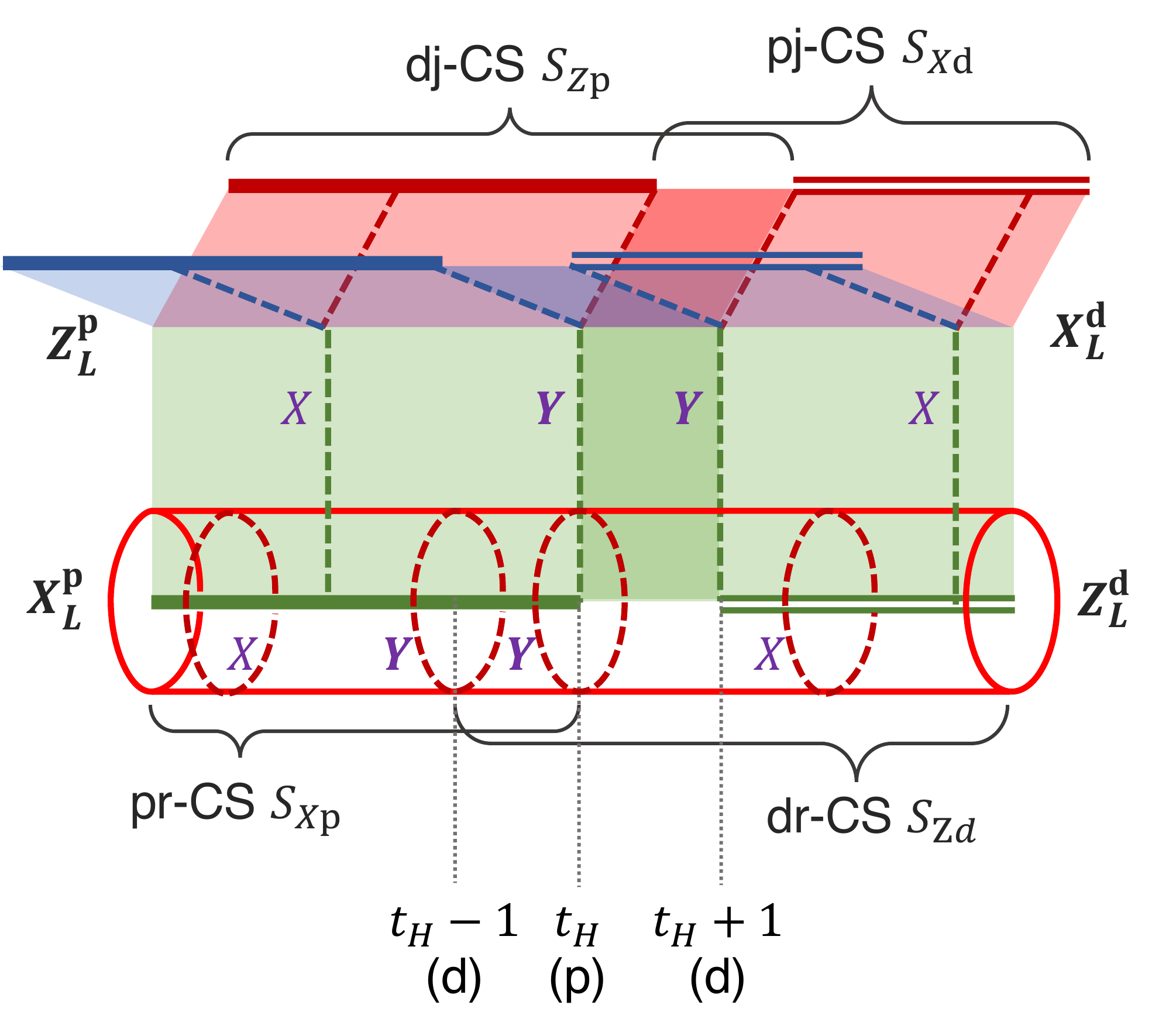}
	\caption{
	Construction of the Hadamard gate from a primal logical qubit to a dual one. 
	Each colored single (double) line is the primal (dual) defect of that color.
	The presented \tsf{dj-CS} $S_{Z\tsf{p}}$ and \tsf{pj-CS} $S_{X\tsf{d}}$ end at the three primal or dual defects and the $\qty(t_H+1)$- or $t_H$-layer, respectively.
	The presented \tsf{pr-CS} $S_{X\tsf{p}}$ and \tsf{dr-CS} $S_{Z\tsf{d}}$ surround the \tsf{pg-D} or \tsf{dg-D} and end at the $t_H$- or $\qty(t_H-1)$-layer, respectively.
    $S_{ZX} := S_{Z\tsf{p}} S_{X\tsf{p}}$ and $S_{XZ} := S_{X\tsf{p}} S_{Z\tsf{d}}$ transform the logical Pauli operators as Eq.~\eqref{eq:hadamard_gate_transformation}.
	The supports of $S_{ZX}$ and $S_{XZ}$ are marked as colored dashed lines.
	For the $Y$-measurements to be fault-tolerant, a dual Y-plane is placed on the $t_H$-layer and primal Y-planes are placed on the $\qty(t_H-1)$- and $\qty(t_H+1)$-layer.
	}
	\label{fig:hadamard_gate}
\end{figure}

If the Hadamard gate is located in the middle of the circuit, it is a bit tricky.
Since $X_L$ and $Z_L$ of a logical qubit can be connected only with primal or dual \tsf{CS}s, respectively, there should be a \tsf{CS} having different primalities near the input and output layers, to achieve the transformation.
To solve this problem, we construct a defect structure starting with a primal logical qubit and ending with a dual one as shown in Fig.~\ref{fig:hadamard_gate}, where the primal one stops at the primal $t_H$-layer and the dual one starts from the dual $\qty(t_H+1)$-layer.
Each pair of defects with the same color must have exactly the same spatial structure at $t=t_H$ and $t=t_H+1$.
Note that such a configuration is possible thanks to the self-duality of the 2D color codes which makes primal and dual layers have exactly the same structure.

We consider two pairs of overlapping primal and dual \tsf{CS}s: $\qty( S_{Z\tsf{p}}, S_{X\tsf{d}} )$ and $\qty( S_{X\tsf{p}}, S_{Z\tsf{d}} )$, where $S_{X\tsf{p}}$, $S_{Z\tsf{p}}$, $S_{X\tsf{d}}$, and $S_{Z\tsf{d}}$ are a \tsf{pr-CS}, \tsf{dj-CS}, \tsf{pj-CS}, and \tsf{dr-CS} defined in Fig.~\ref{fig:hadamard_gate}, respectively.
$S_{ZX} := S_{Z\tsf{p}} S_{X\tsf{d}}$ then transforms $Z_L$ of the input primal logical qubit to $X_L$ of the output dual one.
Similarly, $S_{XZ} := S_{X\tsf{p}} S_{Z\tsf{d}}$ transforms the input $X_L$ to the output $Z_L$.
Condition \ref{cond:condition_1} in Sec.~\ref{subsubsec:identity_gate} is thus satisfied with these two ``hybrid'' \tsf{CS}s.
What remains is Condition \ref{cond:condition_2}.
Since $S_{ZX}$ and $S_{XZ}$ contain $Y$ operators on some \tsf{CQ}s in the overlapping regions, the qubits should be measured in the $Y$ basis for the \tsf{CS}s to be compatible.

To make the $Y$-measurements fault-tolerant, we introduce \textit{Y-planes}:

\begin{definition}[\textbf{Y-plane}]
    A \textup{primal (dual) Y-plane} is the set of \tsf{p(d)CQ}s in a continuous area contained in a dual (primal) layer.
    \tsf{CQ}s in Y-planes are measured in the $Y$ basis.
\end{definition}

Errors in Y-planes can be corrected by an error correction procedure presented in Sec.~\ref{subsec:measurement_pattern_and_defects}.
Therefore, the $Y$-measurements for the Hadamard gate can be fault-tolerantly done by placing wide enough Y-planes to cover $\supp_Y \qty( S_{ZX} )$ and $\supp_Y \qty( S_{XZ})$ completely.



\subsubsection{Phase gate}
\label{subsubsec:phase_gate}

\begin{figure}[t!]
    \centering
    \includegraphics[width=0.6\columnwidth]{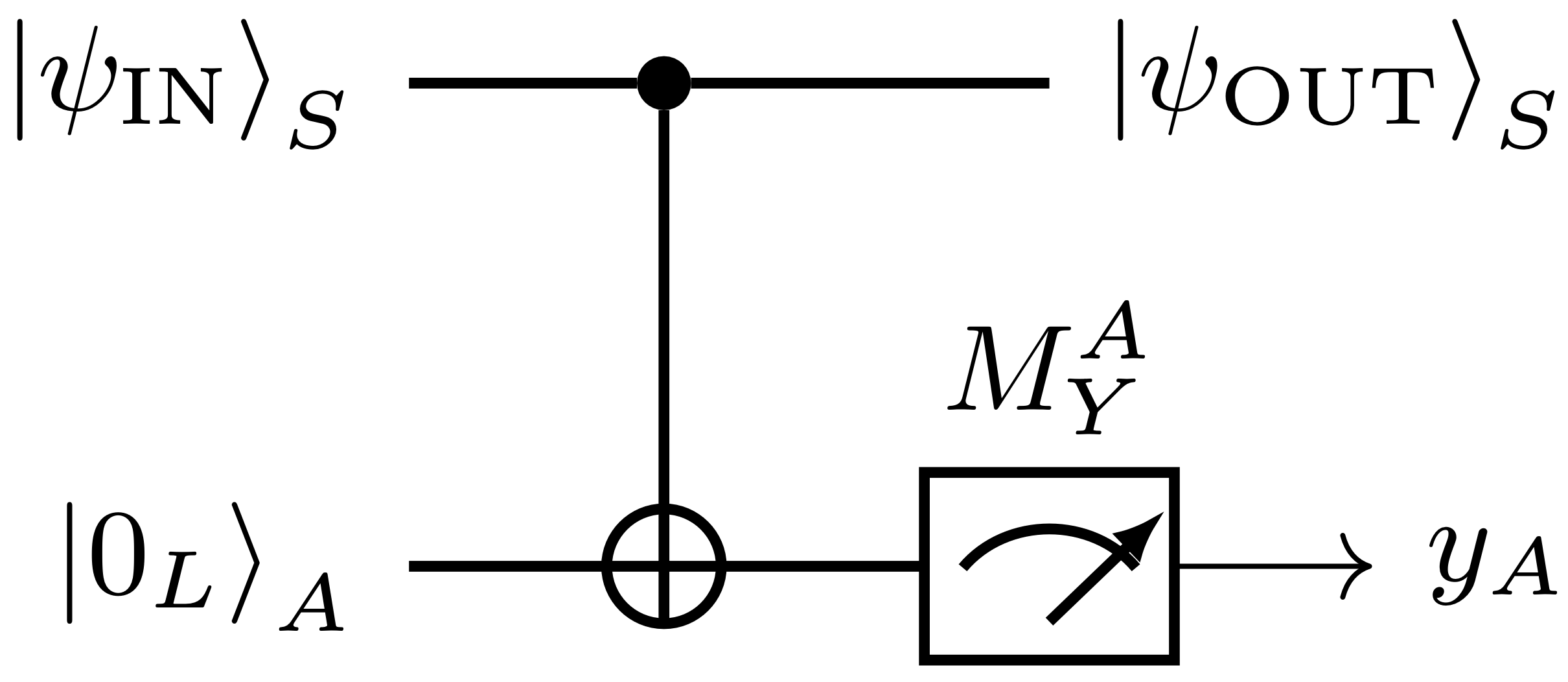}
    \caption{Circuit for the logical phase gate of a system logical qubit ($S$) with an ancilla logical qubit ($A$). $M_Y^A$ is the $Y_L$-measurement on the ancilla qubit with the result of $y_A$.}
    \label{fig:phase_gate_circuit}
\end{figure}

We now complete the generating set of the Clifford group with the construction of the logical phase ($S_L$) gate.
The phase gate is achieved indirectly by utilizing an ancilla logical qubit; the circuit in Fig.~\ref{fig:phase_gate_circuit} implements $S_L$ if the $Y_L$-measurement of the ancilla logical qubit gives the result of $+1$ and $S_L Z_L$ if the result is $y_A = -1$.

\begin{figure*}[t!]
	\centering
	\includegraphics[width=\textwidth]{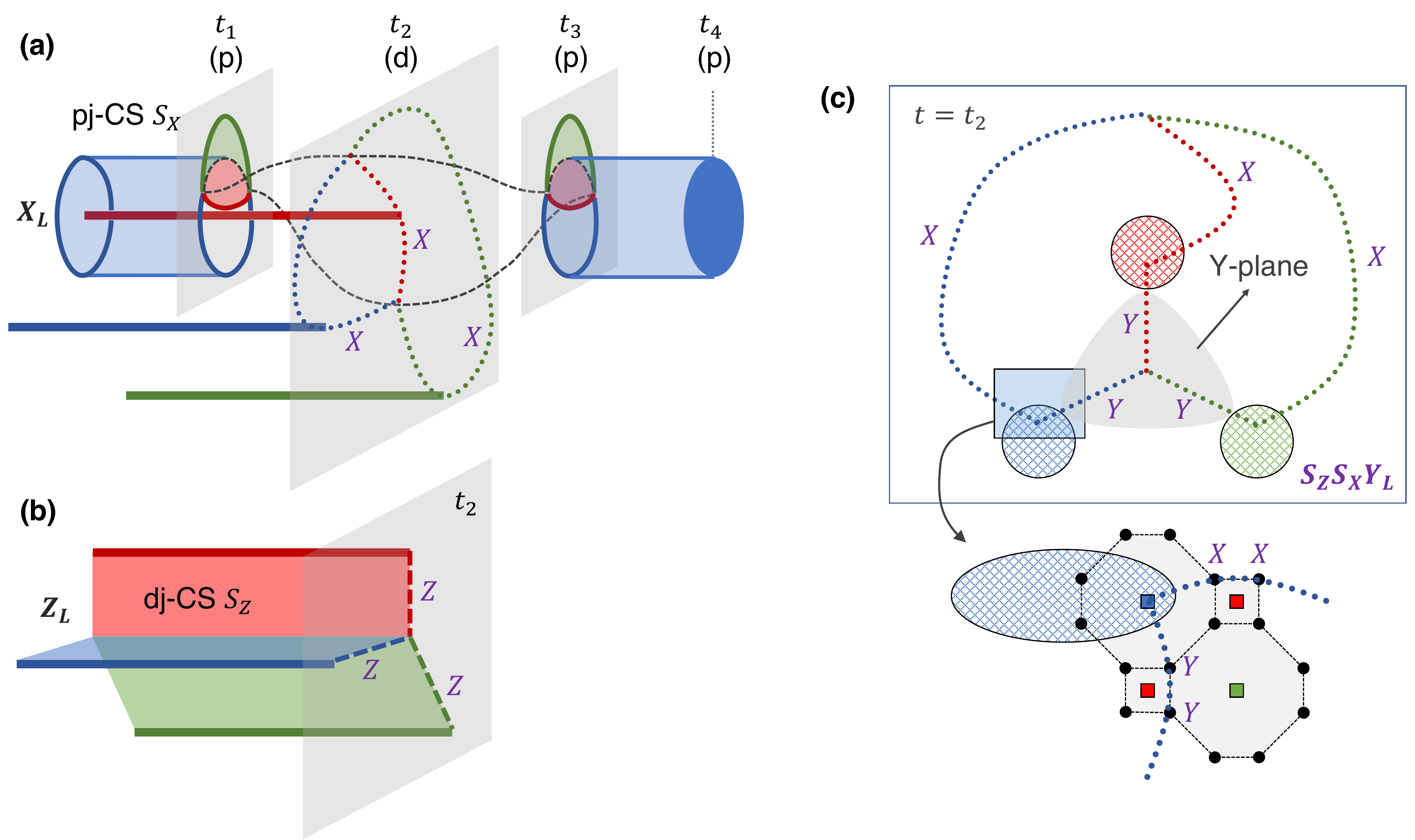}
	\caption{
	(a) \tsf{pj-CS} $S_X$ and (b) \tsf{dj-CS} $S_Z$ for the $Y_L$-measurement of a primal logical qubit.
	Near the input layer, $S_X$ connected with $X_L$ of the input logical qubit has the form of a \tsf{pb-CS} surrounding the \tsf{pr-D}.
	On the $t_1$-layer, it is divided into three \tsf{CS}s with different colors through a spacelike joint.
	Each \tsf{CS} is then deformed appropriately so that the joint is extended along the black dashed line.
	The 1-chains (colored dotted lines) on the $t_2$-layer are in the $X$-support of $S_X$.
	On the $t_3$-layer, the joint becomes spacelike again.
	After that, $S_X$ has the form of a \tsf{pb-CS} until the $t_4$-layer on which $S_X$ is closed.
	$S_Z$ simply connects $Z_L$ and the 1-chains (colored dashed lines) on the $t_2$-layer.
	(c) $\supp\qty(S_Z S_X Y_L)$ on the $t_2$-layer.
	A Y-plane is placed to cover $\supp_Y \qty(S_Z S_X Y_L)$.
	The junction of the \tsf{pb-D} and the blue 1-chain is explicitly shown below.
	}
	\label{fig:logical_Y_measurement}
\end{figure*}

All the elements in the circuit already have been described except the $Y_L$-measurement.
For the $Y_L$-measurement of a input primal logical qubit, we extend the defects straight along the time axis to a layer ($t_2$).
As shown in Fig.~\ref{fig:logical_Y_measurement}, there exists a \tsf{pj-CS} $S_X$ (connected with $X_L$) and \tsf{dj-CS} $S_Z$ (connected with $Z_L$) which are stabilizers, such that $S_Z S_X Y_L = i \qty(S_Z Z_L)\qty( S_X X_L)$ contains $X$ operators on vacuum qubits and $Y$ operators along ``Y''-shaped connected 1-chains on the $t_2$-layer.
Hence, the $Y_L$-measurement can be done by placing a Y-plane on the $t_2$-layer as Fig.~\ref{fig:logical_Y_measurement}(c).

We conclude that all the elements in the circuit of Fig.~\ref{fig:phase_gate_circuit} can be implemented fault-tolerantly, thus the fault-tolerant phase gate can be made up to byproduct operators.

\subsection{State injection}
\label{subsec:state_injection}

Preparation of an arbitrary logical qubit $a \ket{0_L} + b \ket{1_L}$ is essential for implementing the logical $T$ gate as well as quantum computation with arbitrary input states.
This is done in our scheme by injecting the corresponding unencoded state into a physical qubit.

\begin{figure}[t!]
	\centering
	\includegraphics[width=\columnwidth]{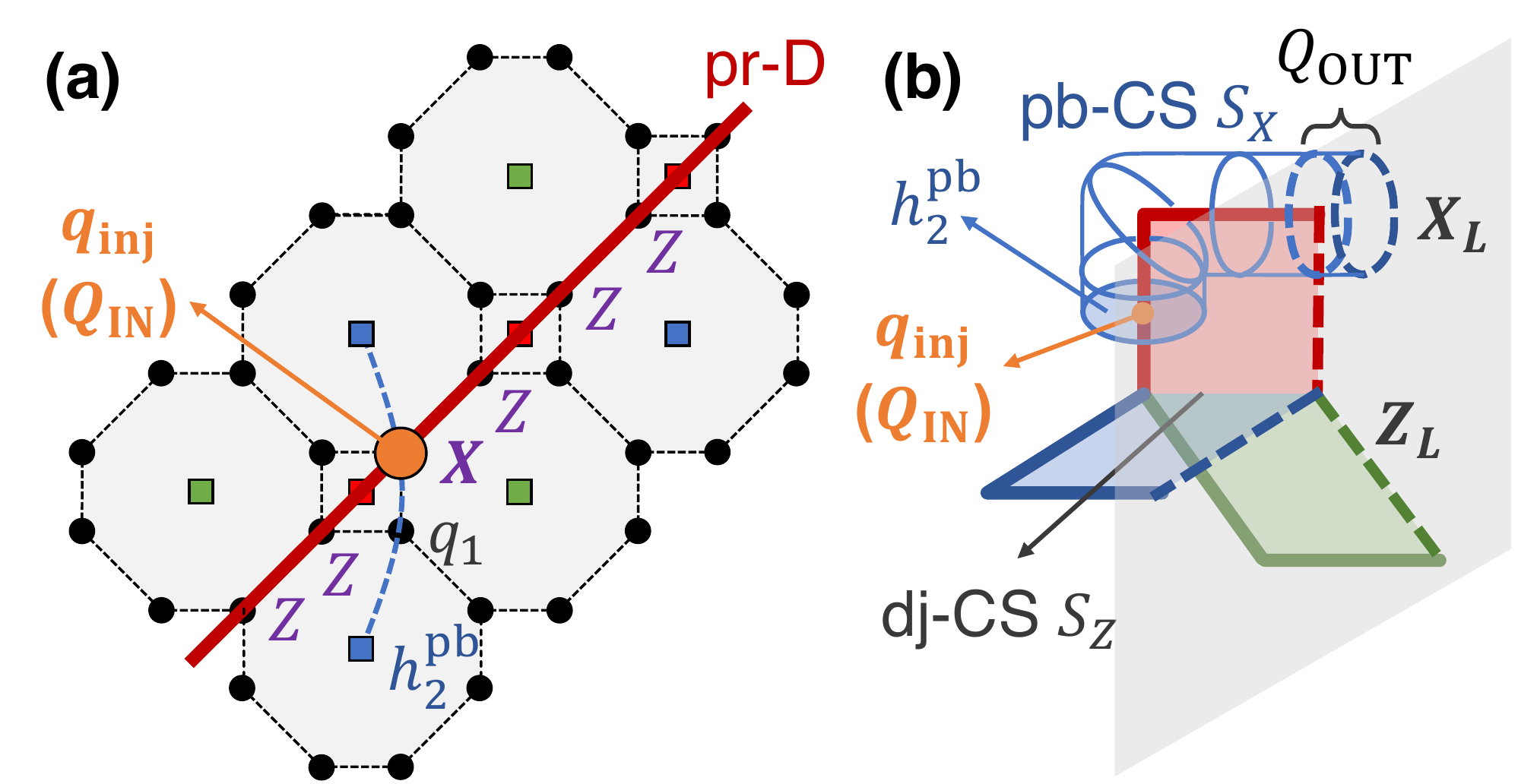}
	\caption{
	State injection procedure.
	(a) An unencoded state is injected into an injection qubit $q_\mathrm{inj}$, which is the only input qubit, in the \tsf{pr-D} which is spacelike and thicknessless at $q_\mathrm{inj}$.
	$Z\qty(q_\mathrm{inj})$ is invariant when the \cz~gates associated with $q_\mathrm{inj}$ are applied.
	However, $X\qty(q_\mathrm{inj})$ is transformed into $S\qty(q_\mathrm{inj})$, where $S\qty(q_\mathrm{inj})$ is the C-type SG around $q_\mathrm{inj}$.
	$S\qty(q_\mathrm{inj})$ is equivalent to $S_\tsf{CS}\qty(h_2^\tsf{pb})$ since $S_\tsf{CS}\qty(h_2^\tsf{pb}) = S\qty(q_\mathrm{inj}) S\qty(q_1)$, where $h_2^\tsf{pb} \in H_2^\tsf{pb}$ is the timelike 2-chain marked as a blue dashed line and $q_1$ is the marked \tsf{CQ} adjacent to $q_\mathrm{inj}$.
	$q_\mathrm{inj}$ is measured in the $X$ basis during the measurement step.
	(b) $S_\tsf{CS}\qty(h_2^\tsf{db})$ is transformed into $X_L$ of the output logical qubit via the \tsf{pb-CS} $S_X$.
	$Z\qty(q_\mathrm{inj})$ is transformed into $Z_L$ of the output logical qubit via the \tsf{dj-CS} $S_Z$.
	}
	\label{fig:state_injection}
\end{figure}

\begin{figure*}[t!]
	\centering
	\includegraphics[width=\textwidth]{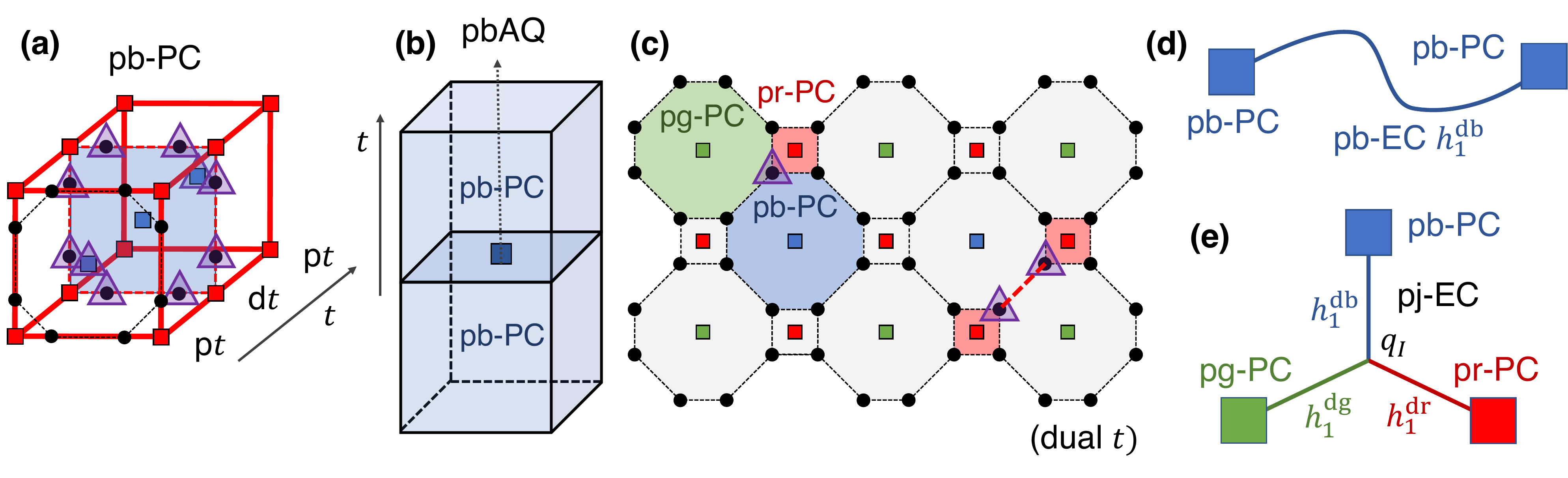}
	\caption{
	(a) Explicit structure of a parity-check operator (\tsf{PC}), specifically a \tsf{pb-PC} in a 4-8-8 CCCS. 
	Purple triangles indicate its $X$-support qubits.
	(b) A $Z$ or $X$-measurement ($M_X$) error on a \tsf{pcAQ} flips two \tsf{pc-PC}s sandwiching $q$.
	(c) A dual layer of a 4-8-8 CCCS is presented.
	Purple triangles indicate the \tsf{pCQ}s with errors.
	Each \tsf{c}-colored face corresponds to a flipped \tsf{pc-PC}, where an example is shown in (a) as a blue face on the dual layer.
	(d) A primal blue error chain (\tsf{pb-EC}), where every qubit along a connected dual 1-chain $h_1^\tsf{db}$ has an error, flips two \tsf{pb-PC}s located at its two ends.
	(e) Starting from an error on a \tsf{pCQ} $q_I$, a \tsf{pj-EC} is constructed by multiplying a \tsf{pc-EC} ending at the flipped \tsf{pc-PC} for each color \tsf{c} to the error operator.
	A \tsf{pj-EC} flips three primal \tsf{PC}s located at its ends.
	}
	\label{fig:vacuum_error_correction}
\end{figure*}

We start from the configuration for the $Z_L$-initialization of a primal logical qubit shown in Fig.~\ref{fig:logical_qubit}(c), where three defects meet at a point.
First, a qubit $q_\mathrm{inj}$ in the \tsf{pc-D} for any color \tsf{c} is selected as an \textit{injection qubit} which is the only input qubit in $Q_\mathrm{IN}$.
We assume that the defect is ``thicknessless'' at $q_\mathrm{inj}$; namely, its cross-section at $q_\mathrm{inj}$ contains at most one qubit as shown in Fig.~\ref{fig:state_injection}(a).
The desired initial state is injected into $q_\mathrm{inj}$ in an unencoded form $\ket{\psi} = a\ket{0} + b\ket{1}$, then the associated \cz~gates are applied.
Remark that $q_\mathrm{inj}$ is measured in the $X$ basis as stated in Eq.~\eqref{eq:measurement_pattern}.
The $X$ ($Z$) operator on $q_\mathrm{inj}$ is transformed into $X_L$ ($Z_L$) up to a sign factor as shown in Fig.~\ref{fig:state_injection}, thus the logical state $\ket{\psi_L} = a\ket{0_L} + b\ket{1_L}$ is prepared up to byproduct operators.

Note that the state injection procedure is inherently not fault-tolerant, since it uses an unprotected single-qubit state and the defect is thicknessless at $q_\mathrm{inj}$.
Therefore, magic state distillation is essential for the faithful $T$ gate.

\section{Error correction}
\label{sec:error_correction}

Now we describe error correction schemes in CCCSs.
The scheme varies with the area of the qubits: the vacuum, defects, and Y-planes.

\subsection{Error correction in the vacuum and defects}
\label{subsec:error_correction_vacuum}

For error correction in the vacuum, we exploit \textit{parity-check operators} (\tsf{PC}s) defined as follows:

\begin{definition}[\textbf{Parity-check operator}]
    For each cell $c$, the \tsf{CS}
    \begin{align*}
        S_\tsf{CS}(\partial c) = X(\partial c)
    \end{align*}
    is a \textup{parity-check operator (\tsf{PC})}, where $S_\tsf{CS}(\cdot)$ is given in Eq.~\eqref{eq:correlation_surface_definition}.
\end{definition}

\tsf{PC}s are classified into six groups according to primalities and cell colors. 
Here, the primality of a \tsf{PC} $S_\tsf{CS}\qty(\partial c)$ is that of the shrunk lattice $\mathcal{L}$ containing the cell $c$, and its cell color is the color of the \tsf{AQ} $Q\qty(c)$.
Remark that the cell color is different from the color of $\mathcal{L}$, as shown in Table~\ref{table:qubits_corr_to_each_element_in_shrunk_lattices}.
We refer to a primal \tsf{c}-colored \tsf{PC} as a ``\tsf{pc-PC}.''

Remark that a given \tsf{dcAQ} $q$ corresponds to two primal cells, one for each of $\mathcal{L}^{\tsf{pc}_1}$ and $\mathcal{L}^{\tsf{pc}_2}$ where \tsf{c}, $\tsf{c}_1$, and $\tsf{c}_2$ are all different colors.
However, the \tsf{PC}s corresponding to the cells are indeed the same, comparing Fig.~\ref{fig:CCCS_shrunk_lattices}(a) and (b) as an example.
We can thus regard that one \tsf{AQ} ($q$) corresponds to one \tsf{PC}, and denote it as $S_\tsf{PC}(q)$.
The support of the \tsf{pc-PC} $S_\tsf{PC}(q)$ for a \tsf{dcAQ} $q$ contains two \tsf{pcAQ}s and multiple \tsf{pCQ}s around $q$ as shown in Fig.~\ref{fig:vacuum_error_correction}(a)

We first consider only vacuum qubits.
Since they are measured in the $X$ basis, all \tsf{PC}s survive as stabilizers after the measurement step.
Any $Z$ error before the measurement or any $X$-measurement ($M_X$) error flips several \tsf{PC} outcomes.
Note that $X$ errors do not affect the outcomes at all, so can be ignored.
The final step for error correction is to decode errors from them and correct the errors.

An error may occur on either an \tsf{AQ} or a \tsf{CQ}. 
An error on a \tsf{pcAQ} $q$ flips two \tsf{pc-PC}s sandwiching $q$ along the time axis as shown in Fig.~\ref{fig:vacuum_error_correction}(b).
An error on a \tsf{pCQ} $q$ flips \tsf{pr-PC}, \tsf{pg-PC}, and \tsf{pb-PC} surrounding $q$ spatially, as shown in Fig.~\ref{fig:vacuum_error_correction}(c).
If both the \tsf{pCQ}s constituting a \tsf{pcL} $l$ have errors, the two \tsf{pc-PC}s connected by $l$ are flipped.

Combining the above facts, we conclude that, if every qubit in $Q(h_1^\tsf{dc})$ for a connected dual 1-chain $h_1^\tsf{dc} \in H_1^\tsf{dc}$ has an error, the \tsf{pc-PC} $S_\tsf{PC} (q)$ for each qubit $q \in Q(\partial h_1^\tsf{dc})$ is flipped, as shown in Fig.~\ref{fig:vacuum_error_correction}(d).
Such an error set in the vacuum is called a primal \tsf{c}-colored \textit{error chain}, referred to as a ``\tsf{pc-EC}.''
Formally, a \tsf{pc-EC} is written as the tensor product of the $Z$ operators on the error qubits.
Furthermore, starting from an error on a \tsf{pCQ}, each flipped \tsf{PC} may be ``moved'' by multiplying a primal error chain of the corresponding color ending at the \tsf{PC}.
An error set constructed by this way flips three primal \tsf{PC}s located at its ends and is referred to as a ``\tsf{pj-EC}.''
General error chains are obtained by connecting multiple \tsf{pc-EC}s for each color \tsf{c} and \tsf{pj-EC}s.

\begin{figure}[t!]
	\centering
	\includegraphics[width=\columnwidth]{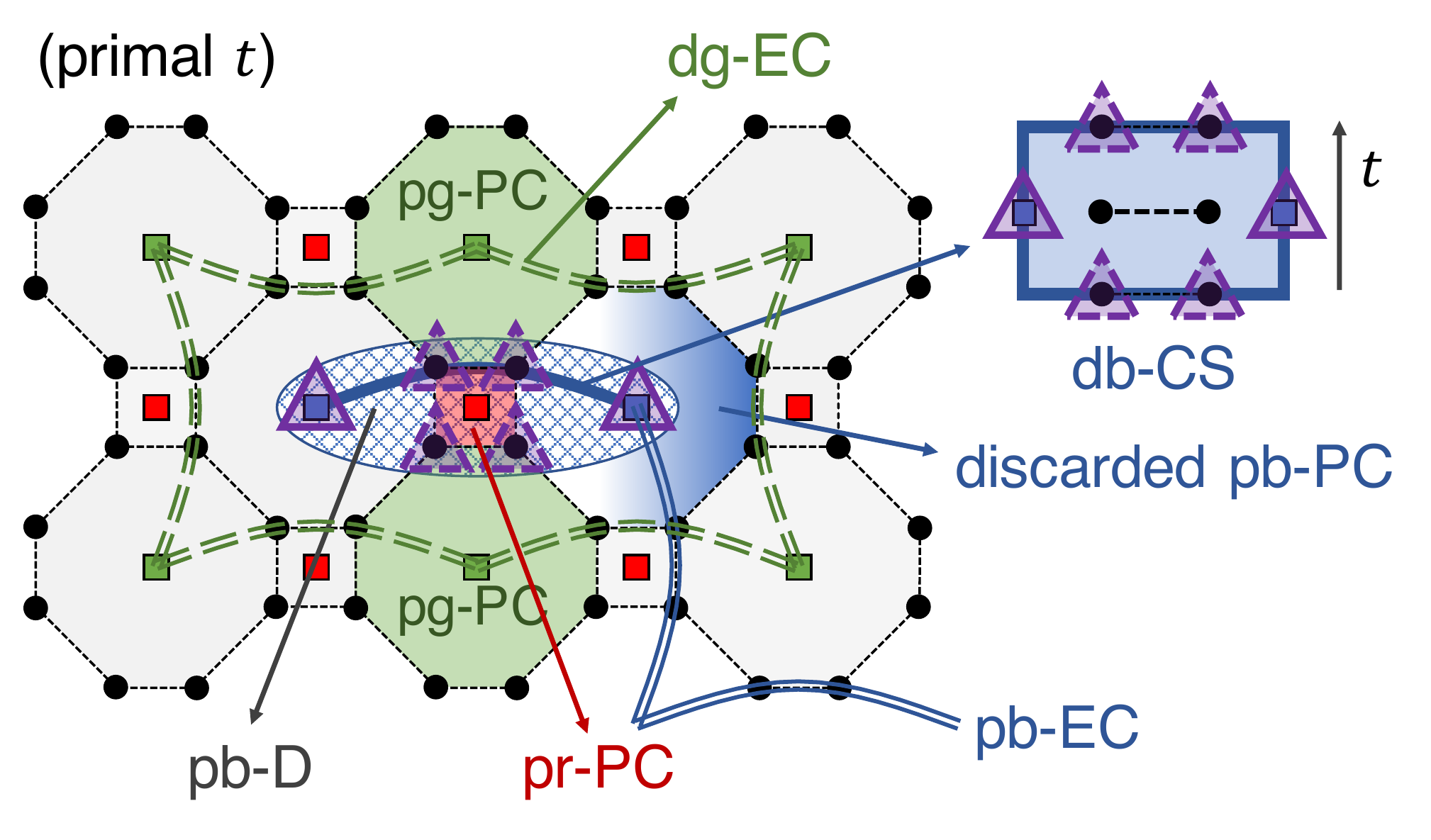}
	\caption{
	    \tsf{PC}s deformed or created due to a timelike \tsf{pb-D} in a 4-8-8 CCCS.
	    A cross-section of the defect on a primal layer is presented.
	    Each purple triangle with a solid (or dashed) border indicates a defect qubit on the layer (or an adjacent dual layer).
	    Although each of the two \tsf{pg-PC}s and one \tsf{pr-PC} marked is incompatible with the defect, their product is compatible, thus it can be used for error correction.
	    However, the marked \tsf{pb-PC} cannot be merged with other \tsf{PC}s in such a way, thus there is no choice but to discard it, which makes a \tsf{pb-EC} ending at it (blue double solid line) undetectable.
	    A \tsf{dg-EC} surrounding the defect (green double dashed line) located in an adjacent dual layer is another nontrivial undetectable set of errors.
	    Some dual \tsf{CS}s including the presented \tsf{db-CS} additionally survive and can be used to detect errors in both the vacuum and defect.
	}
	\label{fig:defect_error_correction}
\end{figure}

We now investigate the effects of a \tsf{pc-D} $d$ to the nearby \tsf{PC}s.
First, all primal \tsf{PC}s whose supports contain any defect qubit no longer survive, while dual \tsf{PC}s are unaffected.
Such incompatible \tsf{PC}s may be multiplied with each others to form larger compatible stabilizers, as shown in Fig.~\ref{fig:defect_error_correction} where two \tsf{pg-PC}s and a \tsf{pr-PC} are merged.
Like normal \tsf{PC}s, these merged \tsf{PC}s also can detect errors, although decoding errors from \tsf{PC}s may get more ambiguous.
Some \tsf{PC}s for which such multiplication is impossible have no choice but to be discarded, as the \tsf{pb-PC} in Fig.~\ref{fig:defect_error_correction}.
As a consequence, a \tsf{pc-EC} ending at $d$ is not detected by any \tsf{PC} near $d$.

A \tsf{pc-D} may make some dual \tsf{CS}s survive additionally.
The \tsf{dc-CS} $S_\tsf{CS}(f)$ for a face $f \in \mathcal{B}_2^\tsf{dc}$, where $Q\qty(\partial f)$ is in the defect as shown in Fig.~\ref{fig:defect_error_correction}, is compatible, thus can serve as a \tsf{PC} for detecting errors.
We call such \tsf{CS}s \textit{defect \tsf{PC}s}.
A notable thing is that they may detect not only errors on vacuum qubits but also $X$ or $Z$-measurement ($M_Z$) errors on defect qubits.

We then identify nontrivial undetectable error chains, where ``nontrivial'' here means that they incur logical errors.
Such an error chain is closed or ends at defects with the same primality and color.
Considering the identity gate of a primal logical qubit, the shortest error chain inducing an $X_L$ error is a \tsf{pj-EC} ending at the three defects, and the shortest one inducing a $Z_L$ error is a closed \tsf{dc-EC} surrounding the $\tsf{pc}'\tsf{-D}$ ($\tsf{c}' \neq \tsf{c}$), as shown schematically in Fig.~\ref{fig:identity_gate} and explicitly in Fig.~\ref{fig:defect_error_correction}.
Note that a closed \tsf{dc-EC} penetrating a defect is detected by defect \tsf{PC}s, although not detected by ordinary \tsf{PC}s.
The \textit{code distance} of a logical qubit is defined by the size of the smallest nontrivial undetectable error set, and it increases as the defects get thicker or get farther from each other \footnote{
    Such an error set may not be an error chain, because errors in defects may also be nontrivial and undetectable. 
    However, since error correction in defects gets more accurate as the defects get thicker, the dependency of the code distance to their thicknesses is still valid. 
    It can be verified that the smallest nontrivial undetectable error set containing only defect qubits completely covers a cross-section of the defect, thus its size is roughly a quadratic function of the circumference of the defect.
}.




\subsection{Error correction in Y-planes}
\label{subsec:error_correction_y_planes}

\begin{figure}[t!]
	\centering
	\includegraphics[width=\columnwidth]{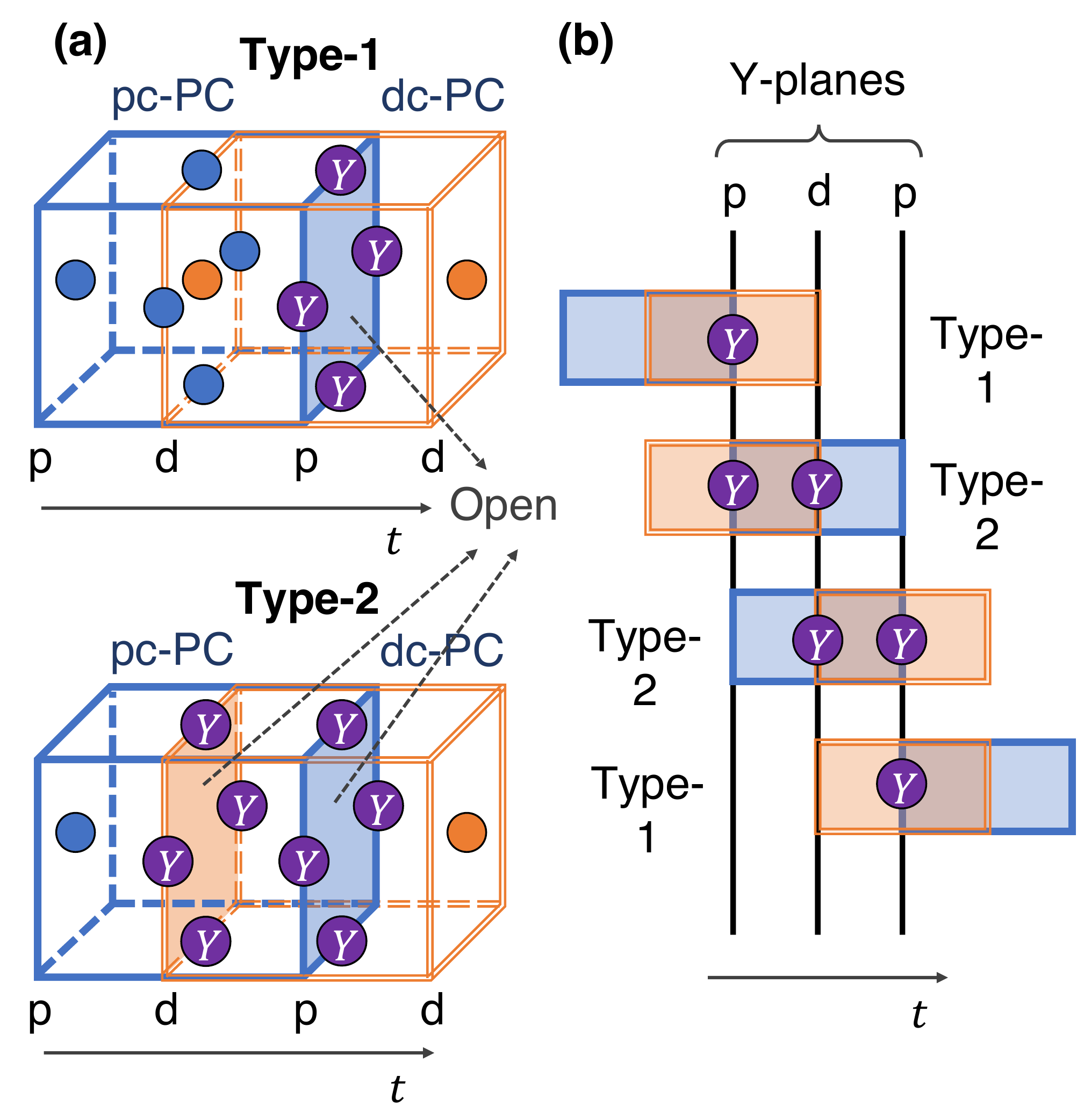}
	\caption{
	(a) Two types of hybrid \tsf{PC}s.
    Each cube is a \tsf{pc-PC} (blue solid line) or \tsf{dc-PC} (orange double line) which constitutes a hybrid \tsf{PC} and may be open, where the colored faces indicate the open faces.
    $Y$-support qubits of the hybrid \tsf{PC}s are marked as purple circles, and their $X$-support qubits originally in the \tsf{p(d)c-PC}s are marked as blue (orange) circles.
    (b) For three consecutive Y-planes, a hybrid \tsf{PC} may be placed in the presented four ways.
    Errors in the Y-planes may be corrected using the outcomes of such hybrid \tsf{PC}s covering the entire Y-planes.
	}
	\label{fig:y_plane_error_correction}
\end{figure}

To correct errors in Y-planes, we use \textit{hybrid \tsf{PC}}s defined with \textit{open \tsf{PC}s}, which are visualized in Fig.~\ref{fig:y_plane_error_correction}(a).

\begin{definition}[\textbf{Open \tsf{PC}}]
    For a cell $c$ and a face $f \in \partial c$, an \textup{open \tsf{PC}} is $S_\tsf{CS}(\partial c)S_\tsf{CS}(f)$, where $f$ determines the direction toward which it is open.
\end{definition}

\begin{definition}[\textbf{Hybrid \tsf{PC}}]
    A \textup{hybrid \tsf{PC}} is the product of primal and dual \tsf{PC}s of the same color which are adjacent along the time axis and may be open toward each other.
    It is \textup{type-1} if one of the two composing \tsf{PC}s is open, while it is \textup{type-2} if both of them are open.
\end{definition}

\begin{figure*}[t!]
	\centering
	\includegraphics[width=\textwidth]{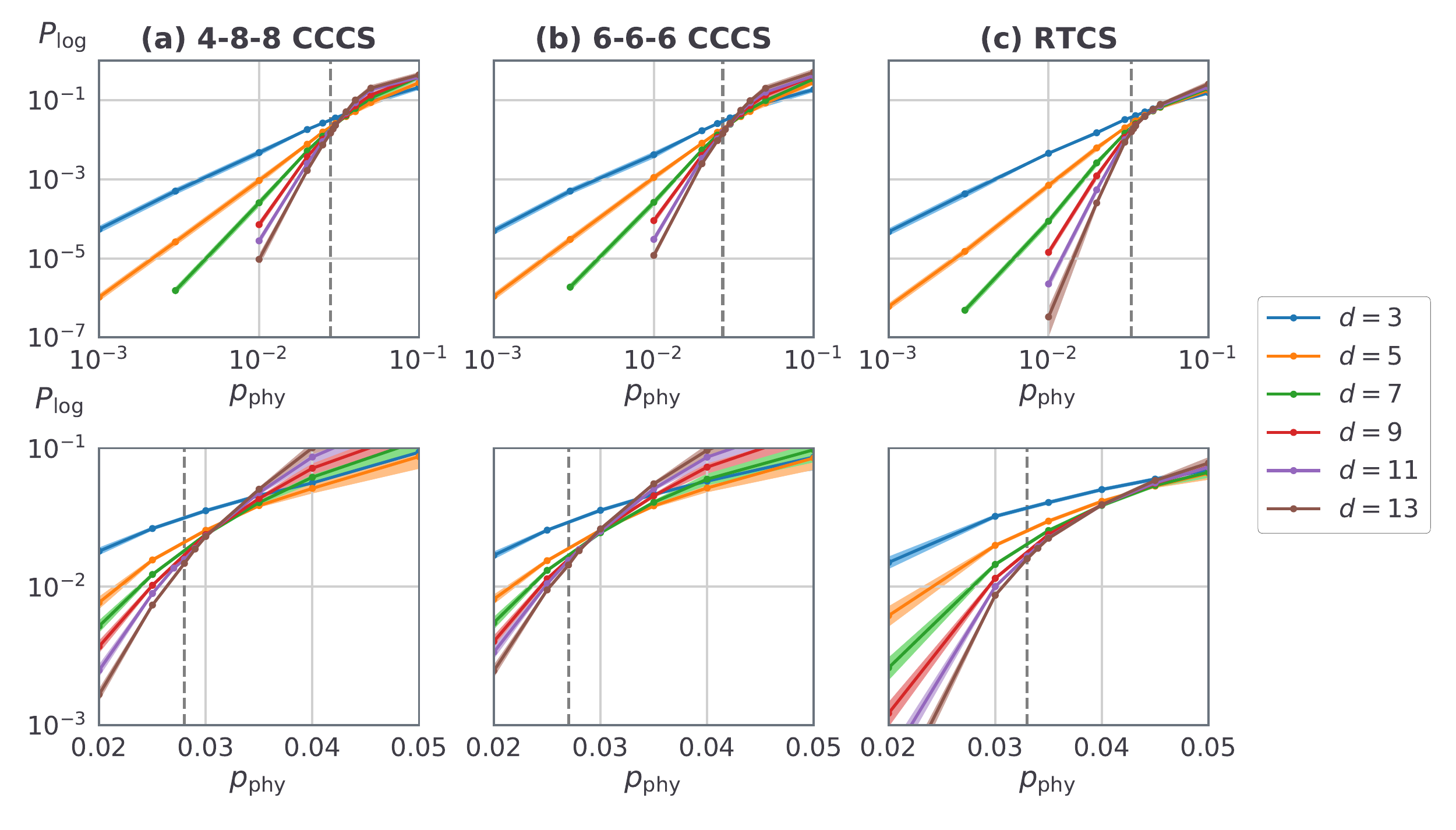}
	\caption{
	    $Z_L$ error probability $P_\mathrm{log}$ versus physical-level error probability $p_\mathrm{phys} = p_Z + p_{M_X} - p_Z p_{M_X}$, where $p_Z$ ($p_{M_X}$) is the $Z$ ($X$-measurement) error probability on vacuum qubits, for different code distances with respect to (a) 4-8-8 CCCSs, (b) 6-6-6 CCCSs, and (c) RTCSs.
	    The graphs in the upper row show the results for $0.001 \leq p_\mathrm{phys} \leq 0.1$, and those in the lower row show the results near the threshold values.
	    Pale areas around the lines indicate the 99\% confidence intervals of $P_\mathrm{log}$.
	    The error thresholds are calculated as 2.7\% for 6-6-6 CCCSs, 2.8\% for 4-8-8 CCCSs, and 3.3\% for RTCSs, which are shown as grey dashed lines.
	}
	\label{fig:error_simulation}
\end{figure*}

Remark that \tsf{CQ}s in Y-planes are measured in the $Y$ basis, thus ordinary \tsf{PC}s whose supports contain those \tsf{CQ}s are incompatible with the qubits.
Instead of them, we use hybrid \tsf{PC}s whose $Y$-supports are on the Y-planes, as shown in Fig.~\ref{fig:y_plane_error_correction}(b).
Error correction in the Y-planes is done with a set of hybrid \tsf{PC}s covering the entire Y-planes.


\section{Calculations}
\label{sec:calculations}

\subsection{Resource overheads}
\label{subsec:resource_overheads}

We now calculate and compare the resource overheads of MBQC via RTCSs or CCCSs.
For each case, we consider a periodic hexagonal arrangement of parallel timelike primal defects, where primal logical qubits with the code distances of $d$ are compactly packed in the space.
In other words, the intervals of the arrangement are determined to minimize the number of physical qubits per logical qubit while keeping all the possible nontrivial undetectable error chains to contain $d$ or more qubits.
We present such arrangements of defects in Appendix~\ref{app:resource_overheads}.

\begin{table}[b!]
    \caption{
        Resource overheads of MBQC via RTCSs or CCCSs, evaluated by the numbers of physical qubits ($n$) and \cz~gates ($N_\mathrm{\cz}$) per layer in terms of the code distance ($d$) and the number of logical qubits ($k$), regarding optimal hexagonal arrangements of parallel timelike primal defects.
        Only the leading-order terms on $d$ are presented.
        Two types of color codes are considered: the 4-8-8 and 6-6-6 lattices.
    }
    \label{table:resource_overheads}
    \centering
    \begin{ruledtabular}
    \begin{tabular}{ccc}
        Types of cluster states   & $n/k$ & $N_\mathrm{\cz}
        /k$ \\ \hline
        RTCS           & $\approx 6.6 d^2$ & $\approx 13.1 d^2$ \\
        4-8-8 CCCS & $\approx 3.9 d^2$ & $\approx 10.5 d^2$ \\
        6-6-6 CCCS & $\approx 3.7 d^2$ & $\approx 9.8 d^2$
    \end{tabular}
    \end{ruledtabular}
\end{table}

Table~\ref{table:resource_overheads} shows the calculated numbers of physical qubits ($n$) and \cz~gates ($N_\mathrm{\cz}$) per layer in terms of $d$ and the number of logical qubits ($k$), considering the optimal hexagonal arrangements.
It is worth noticing that MBQC via CCCSs is definitely more resource-efficient than MBQC via RTCSs; $n/k$ is about 1.7--1.8 times smaller for CCCSs than for RTCSs.
Note that the compact packing of logical qubits may be unrealistic; extra spaces may be needed for implementing logical gates except the identity gate.

\subsection{Error thresholds}
\label{subsec:error_thresholds}

We numerically calculate and compare \textit{error thresholds} of MBQC via RTCSs and CCCSs.

\subsubsection{Error model}

We assume a simple error model where vacuum qubits have $Z$ ($M_X$) errors independently with the same probability $p_Z$ ($p_{M_X}$).
Since a $Z$ error just before the measurement and an $M_X$ error have the same effect, it is enough to consider the net error probability $p_\mathrm{phy} := p_Z + p_{M_X} - p_Z p_{M_X}$.
Note that $X$ errors on the vacuum qubits do not affect the $X$-measurement results at all, thus we neglect them.

\subsubsection{Simulation methods}
For each simulation with a code distance of $d$, we consider the logical identity gate of a primal logical qubit covering consecutive $2T+1$ layers with $T = 4d + 1$ starting from a primal layer.
Simplified defect models presented in Appendix \ref{subapp:simplified_defect_models} are used, instead of considering big areas containing the entire defects.
We calculate the $Z_L$ error probability per layer with the Monte Carlo method; we repeat a sampling cycle many times enough to obtain a desired confidence interval of the $Z_L$ error probability.
Each cycle is structured as follows.

We first prepare a cluster state whose shape and size are determined by $d$ and $T$.
Here we assume perfect preparation, namely, no qubit losses or failures of \cz~gates.
Errors are then randomly assigned to primal qubits with a given probability $p_\mathrm{phy}$, except those in the first and final layers to prevent error chains ending at these layers.
After that, the outcomes of primal \tsf{PC}s are calculated, then decoded to locate errors.
Edmonds' minimum-weight perfect matching (MWPM) algorithm \cite{edmonds1965paths, edmonds1965maximum, fowler2015minimum} via Blossom V software \cite{kolmogorov2009blossom} is used for decoding (once for RTCSs and six times for CCCSs), where the details are presented in Appendix \ref{subapp:decoding_methods}.
We then identify primal error chains connecting different defects which incur $Z_L$ errors by comparing the assigned and decoded errors.
We count such error chains while repeating the cycles and obtain the $Z_L$ error probability per layer $P_\mathrm{log}$.
The error threshold $p_\mathrm{thrs}$ is obtained from the calculated $P_\mathrm{log}$ results for different values of $d$ and $p_\mathrm{phy}$; $P_\mathrm{log}$ decreases as $d$ increases if $p_\mathrm{phy} < p_\mathrm{thrs}$ and vice versa otherwise.

\subsubsection{Results}
Figure \ref{fig:error_simulation} shows the results of the simulations.
The obtained error thresholds are $p_\mathrm{thrs} \approx 3.0\%$ for 4-8-8 and 6-6-6 CCCSs and $p_\mathrm{thrs} \approx 3.5\%$ for RTCSs.
The values for CCCSs are slightly lower than the value for RTCSs, but they have similar orders of magnitude.

\section{Remarks}
\label{sec:conclusion}

In this paper, we have proposed a new topological measurement-based quantum computation (MBQC) scheme via color-code-based cluster states (CCCSs).
We have shown that our scheme is comparable with or even better than the conventional scheme via Raussendorf's 3D cluster states (RTCSs) \cite{raussendorf2006fault, raussendorf2007fault, raussendorf2007topological, fowler2009topological}, in the three aspects mentioned at the very beginning:
\begin{enumerate}
    \item
        \textit{Universality.}
        Initialization and measurements of logical qubits and all the elementary logical gates constituting a universal set of gates (\cnot, Hadamard, phase, and $T$ gates) can be implemented via appropriate placement of defects and Y-planes.
        We described each one of them explicitly in Sec.~\ref{sec:MBQC_via_color_code_based_cluster_states}.
    \item
        \textit{Fault-tolerance.}
        We suggested the error correction scheme for each area of qubits in Sec.~\ref{sec:error_correction}.
        We further verified in Sec.~\ref{subsec:error_thresholds} that the error thresholds for $Z$ or $X$-measurement errors have a similar order of magnitude comparing with the value for RTCSs.
    \item
        \textit{Resource-efficiency.}
        Contrary to the case of using RTCSs, the Hadamard and phase gates do not require state distillation, which typically consumes many ancillary logical qubits \cite{bravyi2005universal, raussendorf2007topological, fowler2009high}, as shown in Sec.~\ref{subsec:elementary_gates}, thanks to the nature of the self-duality of the 2D color codes.
        Moreover, we found out in Sec.~\ref{subsec:resource_overheads} that the minimal number of physical qubits per logical qubit in our scheme is about 1.7--1.8 times smaller than the value for RTCSs.
        As a consequence, MBQC via CCCSs requires a significantly smaller amount of resources than MBQC via RTCSs.
\end{enumerate}

We particularly emphasize the last aspect on resource-efficiency as a definite improvement from the previous schemes, which makes our scheme a more easy-to-implement alternative to those.

Our work has several limitations. 
First, the logical $T$ gate still needs costly state distillation.
Some methods to significantly reduce the cost of distillation have been proposed, such as using logical qubits with low code distances as ancilla qubits \cite{litinski2019magic} or exploiting redundant ancilla encoding and flag qubits \cite{chamberland2020very}.
Moreover, 3D gauge color codes \cite{bombin2007topological, bombin2007exact, bombin2015gauge, kubica2015universal, watson2015qudit, kubica2018three, bombin20182d, bombin2018transversal} enables the implementation of a universal set of gates without distillation.
It may be possible to translate these protocols to be applicable for our MBQC scheme.
We also assume the perfect preparation of states, which is unrealistic.
It is unclear how much the fault-tolerance gets weaker if we consider qubits losses or failures of \cz~gates, which is particularly related to photon losses in optical systems.
It will be interesting future works to further investigate and resolve these problems.

Lastly, we would like to mention a recent work on a general topological MBQC scheme using the Walker-Wang model for the 3-Fermion anyon theory \cite{roberts2020fermion}.
It provides a general framework on universal QC with defect braiding, which produces the MBQC scheme via RTCSs as an example.
There may be some connections between this work and our scheme, which is worth further investigation.

\section*{Acknowledgments}
This work was supported by the National Research Foundation of Korea (NRF-2019M3E4A1080074, NRF-2020R1A2C1008609, NRF-2020K2A9A1A06102946) via the Institute of Applied Physics at Seoul National University and by the Ministry of Science and ICT, Korea, under the ITRC (Information Technology Research Center) support program (IITP-2020-0-01606) supervised by the IITP (Institute of Information \& Communications Technology Planning \& Evaluation).

\appendix




\section{Verification of Eq. \eqref{eq:identical_exp_val}}
\label{app:trans_log_op}

Here we verify Eq. \eqref{eq:identical_exp_val}:
\begin{align}
    \expval{\widetilde{X}_L}{\psi} = \expval{x_X X_L'}{\psi'}.
    \label{eq:app_identical_exp_val}
\end{align}
We first assume $V_X = \qty{q_0}$ for a qubit $q_0$.
Since there exists a stabilizer $S_0$ anticommuting with $X\qty(q_0)$ before the measurements, $\expval{X\qty(q_0)}{\psi} = 0$.
Thus,
\begin{align*}
    \ket{\psi'} &= \norm{\frac{I + x_{q_0} X\qty(q_0)}{2} \ket{\psi}}^{-1} \frac{I + x_{q_0} X\qty(q_0)}{2} \ket{\psi} \\
    &= \frac{I + x_{q_0} X\qty(q_0)}{\sqrt{2}} \ket{\psi}
\end{align*}
holds.
Therefore,
\begin{align*}
    \expval{x_X X_L'}{\psi'} &= \expval{\qty[I + x_{q_0} X\qty(q_0)] x_{q_0} X_L'}{\psi} \\
    &= x_{q_0} \expval{X'_L}{\psi} + \expval{\widetilde{X}_L}{\psi}
\end{align*}
holds.
Since $\widetilde{X}_L$ commutes with all the stabilizers before the measurements, $S_0$ also anticommutes with $X\qty(q_0) \widetilde{X}_L = X'_L$, thus $\expval{X'_L}{\psi}$ vanishes.
Hence, we get Eq.~\eqref{eq:app_identical_exp_val}.
For an arbitrary $V_X$ with $\abs{V_X} > 1$, we can show Eq.~\eqref{eq:app_identical_exp_val} by simply repeating this process for every qubit in $V_X$.

\section{Details on calculation of resource overheads}
\label{app:resource_overheads}

Here we calculate the resource overheads of MBQC via RTCSs or CCCSs, namely, the numbers of physical qubits ($n$) and required \cz~gates ($N_\mathrm{\cz}$) per layer in terms of the code distance ($d$) and the number of logical qubits ($k$), which are presented in Table ~\ref{table:resource_overheads} and Sec.~\ref{subsec:resource_overheads}.
We consider hexagonal arrangements of parallel timelike primal defects, where every error chain connecting different defects or surrounding a defect has $d$ or more qubits.
We need to find the optimal intervals minimizing $n/k$.

We first define the coordinate systems for the analysis.
The $x$ and $y$ axes are presented in Fig.~\ref{fig:cluster_states}(b) for RTCSs and Fig.~\ref{fig:color_code_lattices} for the two types of CCCSs.
The unit length is the length of a side of a unit cell for RTCSs, the distance between adjacent \tsf{prAQ} and \tsf{pgAQ} for 4-8-8 CCCSs, and half the distance between two adjacent \tsf{AQ}s with the same color for 6-6-6 CCCSs.

\begin{figure}[t!]
	\centering
	\includegraphics[width=\columnwidth]{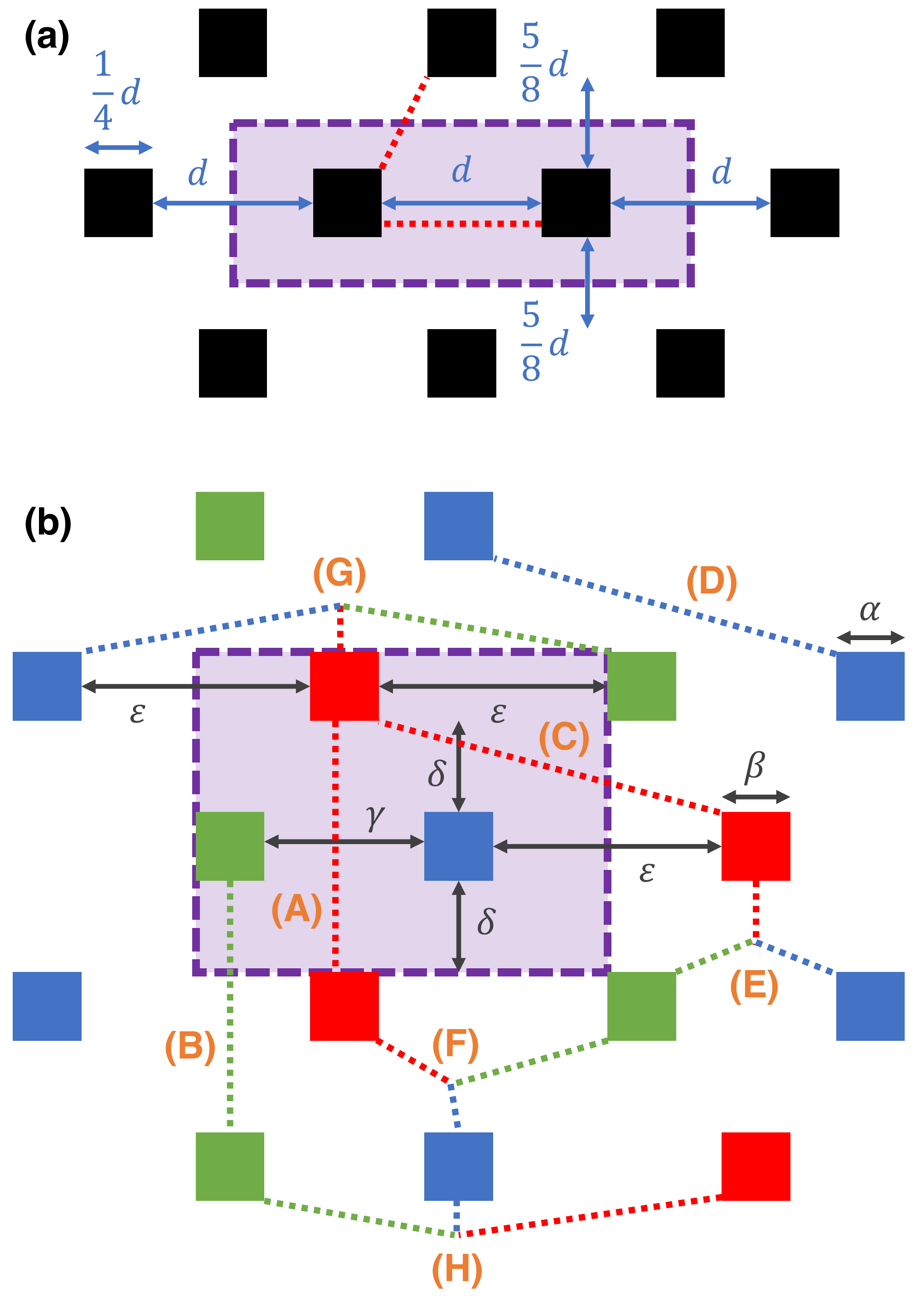}
	\caption{
	    Arrangements of primal defects penetrating a layer for calculating resource overheads of MBQC via (a) RTCSs or (b) CCCSs. 
	    Each black, red, green, or blue square is a defect, where its color means the color of the defect if it is in (b).
	    Each purple rectangle surrounded by dashed lines is an area occupied by a logical qubit.
	    Dotted lines indicate all the possible types of error chains which may be the shortest ones, which are used for obtaining the values of the marked intervals minimizing the area of a logical qubit.
	    Note that, in (b), counterparts of some error chains regarding the exchange of blue and green defects are omitted, since the two lattices (4-8-8 and 6-6-6) which we concern have symmetry on those defects.
	    The optimal intervals for RTCSs are directly presented in (a).
	    For CCCSs, they are $(\alpha, \beta, \gamma, \delta, \epsilon)=(\frac{1}{4} d, \frac{1}{4} d, 0, \frac{1}{2} d, \frac{1}{2} d)$ for 4-8-8 and $(\alpha, \beta, \gamma, \delta, \epsilon) \approx (0.23d, 0.23d, 0.38d, 0.53d, 0.38d)$ for 6-6-6.
	    Here, the unit length is a side of a unit cell in RTCSs (see Fig.~\ref{fig:cluster_states}(b)), the distance between adjacent \tsf{prAQ} and \tsf{pgAQ} in 4-8-8 CCCSs (see Fig.~\ref{fig:CCCS_structure}(a)), and half the distance between two adjacent \tsf{prAQ}s in 6-6-6 CCCSs (see Fig.~\ref{fig:color_code_lattices}(b)).
	}
	\label{fig:resource_overheads_calculation}
\end{figure}

The optimal arrangement in an RTCS is shown in Fig.~\ref{fig:resource_overheads_calculation}(a).
It is straightforward to obtain the intervals, considering that the shortest error chain connecting $(0, 0)$ and $(x, y)$ contains $|x| + |y| + O(1)$ qubits.
Note that we calculate only their leading-order terms on $d$.
The area occupied by a logical qubit is thus about $\frac{35}{16}d^2$, and since a unit area contains three qubits and six \cz~gates, we get $n/k \approx 6.6 d^2$ and $N_\mathrm{\cz} / k \approx 13.1 d^2$.
Note that, for each \tsf{CQ}, we count only one of the two related \cz~gates with other \tsf{CQ}s in the adjacent layers.

It is more tricky to obtain the optimal arrangements in 4-8-8 or 6-6-6 CCCSs.
Figure~\ref{fig:resource_overheads_calculation}(b) shows the concerned hexagonal arrangement with five variables $(\alpha, \beta, \gamma, \delta, \epsilon)$ for the intervals considering the symmetry.
We consider only the leading-order terms of their values on $d$ as well.

We first look at 4-8-8 CCCSs.
The shortest \tsf{pr-EC} connecting $(0, 0)$ and $(x, y)$ contains $2\max(x, y) + O(1)$ qubits, and the shortest \tsf{pg-EC} or \tsf{pb-EC} connecting them contains $|x| + |y| + O(1)$ qubits.
The thicknesses of the defects, $\alpha$ and $\beta$, can be derived from the shortest \tsf{pg-EC} or \tsf{pb-EC} surrounding each defect: $\alpha = \beta = \frac{1}{4} d$.
The following eight inequalities are derived from the eight possible types (A)--(H) of error chain in Fig.~\ref{fig:resource_overheads_calculation}(b):
\begin{align*}
     &(\mathrm{A}) \quad \delta \geq \frac{1}{8} d, \\
     &(\mathrm{B}) \quad \delta \geq \frac{3}{8} d, \\
     &(\mathrm{C}) \quad \max\qty(\gamma + 2\epsilon + \frac{1}{4} d, 2\delta) \geq d, \\
     &(\mathrm{D}) \quad \gamma + 2\delta + 2\epsilon \geq \frac{7}{4} d, \\
     &(\mathrm{E}) \quad \gamma + 2\delta \geq d, \\
     &(\mathrm{F}) \quad \delta + \epsilon + \frac{1}{2} \max\qty( \gamma - \frac{1}{4} d, 0 ) \geq d, \\
     &(\mathrm{G}) \quad \epsilon \geq \frac{3}{8} d, \\
     &(\mathrm{H}) \quad \gamma + 2\epsilon \geq \frac{3}{4} d.
\end{align*}
Note that, to get the inequalities corresponding to (E)--(H), the points at which three error chains meet should be placed carefully.
It is straightforward to see that placing each point just next to the red defect minimizes the length of the error chain.
The area $S$ occupied by a logical qubit is written as
\begin{align}
    S \approx \qty(\alpha + \frac{\beta}{2} + \frac{\gamma}{2} + \epsilon )\qty( 2\delta + \alpha + \beta ).
    \label{eq:area_logical_qubit}
\end{align}
Minimizing $S$ subject to the above inequalities, we get $S \approx \frac{21}{16} d^2$ where the corresponding intervals are $(\alpha, \beta, \gamma, \delta, \epsilon) = \qty(\frac{1}{4} d, \frac{1}{4} d, 0, \frac{1}{2} d, \frac{1}{2} d)$.
A unit area contains three qubits and eight \cz~gates, thus we get $n/k \approx 3.9 d^2$ and $N_\mathrm{\cz}/k \approx 10.5 d^2$.

\begin{figure*}[t!]
	\centering
	\includegraphics[width=\textwidth]{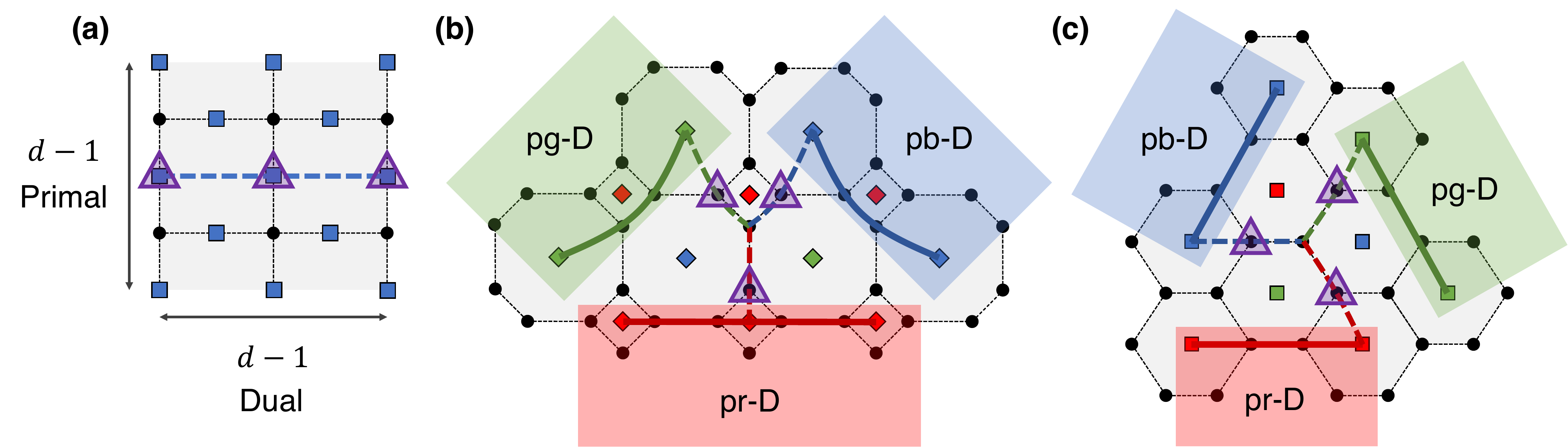}
	\caption{
	Structure of a layer in the simplified defect model for the simulation regarding (a) RTCSs, (b) 4-8-8 CCCSs, or (c) 6-6-6 CCCSs, particularly when the code distance is $d=3$.
	In (a), blue squares (black circles) indicate primal (dual) qubits.
	In (b) and (c), a colored solid line is a boundary corresponding to that color, which can be regarded as a part of a defect.
	For all of them, dashed lines are examples of primal error chains incurring $Z_L$ errors.
	Purples triangles indicate the qubits in the error chains, which show that the code distances are three.
	Defect models for $d > 3$ can be constructed analogously by increasing the distances between the boundaries while keeping their shapes.
	}
	\label{fig:simplified_defects}
\end{figure*}

The optimal arrangement for 6-6-6 CCCSs also can be derived similarly.
The shortest error chain connecting $(0,0)$ and $(x,y)$ for $x,y \geq 0$ contains $\max\qty(x + \frac{1}{\sqrt{3}} y, \frac{2}{\sqrt{3}} y) + O(1)$ qubits.
We thus get $\alpha = \beta = \qty( \sqrt{3} - \frac{3}{2} ) d \approx 0.23 d$, considering an error chain surrounding a defect.
The following inequalities are derived for each type of error chain:
\begin{align*}
     &(\mathrm{A}), (\mathrm{B}) \quad \delta \geq \frac{3 - \sqrt{3}}{4} d, \\
     &(\mathrm{C}), (\mathrm{D}) \quad \max\qty( \frac{1}{2} \gamma + \frac{1}{\sqrt{3}} \delta + \epsilon + \frac{1}{2} \alpha, \frac{2}{\sqrt{3}} \delta ) \geq d, \\
     &(\mathrm{E}) \quad \gamma + \frac{2}{\sqrt{3}} \delta \geq d, \\
     &(\mathrm{F}) \quad \epsilon + \frac{2}{\sqrt{3}} \delta \geq d, \\
     &(\mathrm{G}) \quad 2\epsilon \geq d - \alpha, \\
     &(\mathrm{H}) \quad \epsilon + \gamma \geq d - \alpha.
\end{align*}
Minimizing $S$ in Eq.~\eqref{eq:area_logical_qubit} subject to the inequalities, we get $S \approx 1.42 d^2$ where the corresponding intervals are $(\alpha, \beta, \gamma, \delta, \epsilon) \approx (0.23d, 0.23d, 0.38d, 0.53d, 0.38d)$.
A unit area contains $3\sqrt{3}/2$ qubits and $4\sqrt{3}$ \cz~gates, thus we get $n/k \approx 3.7 d^2$ and $N_\mathrm{\cz}/k \approx 9.8 d^2$.

\section{Details on calculation of error thresholds}
\label{app:error_thresholds}

We here present some details on the calculation of error thresholds presented in Sec.~\ref{subsec:error_thresholds}.

\subsection{Simplified defect models}
\label{subapp:simplified_defect_models}

As mentioned in the main text, we simplify the defect models for efficient simulations.
Instead of considering big regions containing the entire defects, we consider only regions surrounded by boundaries corresponding to the defects.
That is, we only take account of error chains located in the ``inner'' regions surrounded by the defects.
Since those error chains are strictly shorter than error chains passing outside the regions, we conjecture that this assumption does not affect the resulting $Z_L$ error probabilities much.

Figure \ref{fig:simplified_defects} shows single layers of the three simplified defect models for the simulations regarding RTCSs, 4-8-8 CCCSs, and 6-6-6 CCCSs, respectively.
Each layer of the concerned RTCSs has the shape of a square with a side length of $d-1$ in the units of cells for the code distance $d$, where the boundaries are of different types (primal and dual).
Any error chain connecting the two primal boundaries incurs a $Z_L$ error.
For CCCSs, we consider a region surrounded by three boundaries of different colors, where each boundary can be regarded as a part of a defect.
Any error chain connecting the three boundaries incurs a $Z_L$ error.

\subsection{Decoding methods}
\label{subapp:decoding_methods}

\subsubsection{Raussendorf's 3D cluster states}

In an RTCS, the \tsf{PC} outcomes are decoded to locate errors at vacuum qubits via Edmonds' minimum-weight perfect matching algorithm (MWPM) \cite{edmonds1965paths, edmonds1965maximum, fowler2015minimum}, as frequently used in the literature \cite{raussendorf2006fault, barrett2010fault, fowler2012topological, whiteside2014upper}.
Remark that an error chain flips at most two \tsf{PC}s located at its ends, and if it flips one \tsf{PC}, it ends at the boundary.
Hence, our goal is to figure out the most probable set of error chains based on the \tsf{PC} outcomes.

The decoding procedure is briefly summarized as follows.
First, a graph is constructed from the \tsf{PC} outcomes.
The vertex set of the graph contains two vertices for each flipped \tsf{PC}: one is the \tsf{PC} itself and the other is the ``boundary vertex.''
An edge is connected between each pair of different \tsf{PC}s, each pair of a \tsf{PC} and the corresponding boundary vertex, and each pair of different boundary vertices.
A ``weight'' value is assigned to each edge as follows.
If both the vertices are \tsf{PC}s, the weight is the number of qubits in the shortest path between them.
If only one of them is a \tsf{PC}, the weight is the number of qubits in the shortest path between the \tsf{PC} and the closest boundary.
If both of them are boundary vertices, the weight is zero.

We use the MWPM algorithm via Blossom V software \cite{kolmogorov2009blossom} to search for a set of edges of the graph constructed above which covers all the vertices, does not contain duplicated vertices, and minimizes the total weight.
Each edge in the resulting set corresponds to a pair of \tsf{PC}s flipped by an error chain or a \tsf{PC} flipped by an error chain ending at the boundary, unless the edge connects two boundary vertices, which is ignored.
We can thus locate errors from the error chain along the shortest path for each edge.
Since the total weight is minimized, we get the smallest of the sets of edges producing the same \tsf{PC} outcomes, which is the most probable assuming that the error probabilities are independent and the same between qubits.

\subsubsection{Color-code-based cluster states}

The decoding method for RTCSs is not directly applicable to CCCSs, since an error in a CCCS flips at most three \tsf{PC}s, unlike the case of an RTCS.
The decoding for each sample requires the application of the MWPM algorithm six times.

First, the outcomes of \tsf{pb-PC}s and \tsf{pg-PC}s are decoded to find the faces in $\mathcal{L}^\tsf{pr}$ with odd numbers of errors, via the method analogous to that for RTCSs.
This is possible since each of such faces flips at most two (blue or green) \tsf{PC}s like an error in an RTCS.
Remark that each face in $\mathcal{L}^\tsf{pr}$ corresponds to a \tsf{pbAQ}, \tsf{pgAQ}, or \tsf{prL}.
Errors at \tsf{pbAQ}s and \tsf{pgAQ}s are thus obtained from this process.

Next, the left results for \tsf{prL}s and the outcomes of \tsf{pr-PC}s are decoded to locate errors at \tsf{prAQ}s and \tsf{pCQ}s, regarding the parity of the number of errors in each \tsf{prL} as a \tsf{PC}.
This is possible since an error at a \tsf{prAQ} or \tsf{pCQ} flips at most two \tsf{PC}s (\tsf{pr-PC}s and \tsf{prL}s).

All the errors are finally located by the above process.
However, to make the decoding more accurate, we repeat it for $\mathcal{L}^\tsf{pb}$ and $\mathcal{L}^\tsf{pg}$ analogously and select the smallest set of decoded errors among the three results.

\bibliography{references}

\end{document}